\documentclass[12pt]{article}

\usepackage{amsmath,amsthm,amssymb}
\usepackage[linesnumbered,ruled]{algorithm2e}
\usepackage{bm}
\usepackage[usenames]{color}
\usepackage{enumerate}
\usepackage{float}
\usepackage{graphicx,caption,subcaption}
\usepackage{natbib}
\usepackage{booktabs}
\usepackage{fullpage}
\usepackage{longtable}
\usepackage{threeparttablex}

\usepackage[colorlinks=true]{hyperref}
\usepackage{url} 

\usepackage{tikz}
\usetikzlibrary{arrows.meta}

\newtheorem{theorem}{Theorem}
\newtheorem{lemma}{Lemma}

\newtheorem{proposition}{Proposition}

\theoremstyle{definition}

\newtheorem{definition}{Definition}

\newtheorem{remark}{Remark}

\def\spacingset#1{\renewcommand{\baselinestretch}%
{#1}\small\normalsize} \spacingset{1}

\newcommand{\argmin}{\operatorname*{\arg\min}}
\newcommand{\argmax}{\operatorname*{\arg\max}}

\newcommand{\E}{\operatorname{\mathbb{E}}} 
\newcommand{\I}{\operatorname{I}} 

\newcommand{\Cov}{\operatorname{Cov}}

\newcommand{\dto}{\overset{d}{\longrightarrow}}

\newcommand{\pa}{\textnormal{\textsc{pa}}}
\newcommand{\an}{\textnormal{\textsc{an}}}
\newcommand{\ca}{\textnormal{\textsc{ca}}}
\newcommand{\iv}{\textnormal{\textsc{iv}}}
\newcommand{\me}{\textnormal{\textsc{me}}}
\newcommand{\nui}{\textnormal{\textsc{nm}}}

\newcommand{\Diag}{\operatorname{Diag}} 

\newcommand{\inter}{\textsc{in}} 

\newcommand{\bs}[1]{\boldsymbol{\mathbf{#1}}}  

\newcommand{\blind}{0}

\newcommand{\Title}{\Large\bf Discovery and inference of a causal network with hidden confounding}

\begin{document}

\if0\blind
{
    \title{\Title
    \thanks{
        L.~Chen and C.~Li contributed equally. 
        Research is supported in part by NSF grant DMS-1952539, NIH grants R01GM113250, R01GM126002, R01AG065636, R01AG074858, R01AG069895, U01AG073079.
        The authors report there are no competing interests to declare. 
        The authors thank the editor, the associate editor, and three anonymous referees for their helpful comments and suggestions. 
        C.~Li would like to thank R.~Oliver~VandenBerg for the suggestions on writing.
    }}
    \author{Li Chen\footnote{School of Statistics, University of Minnesota, Minneapolis, MN 55455.} 
    \and Chunlin Li\footnote{To whom correspondence should be addressed. Department of Statistics, Iowa State University, Ames, IA 50011. Email: \url{chunlin@iastate.edu}.} 
    \and Xiaotong Shen\footnote{School of Statistics, University of Minnesota, Minneapolis, MN 55455.} 
    \and Wei Pan\footnote{Division of Biostatistics, University of Minnesota, Minneapolis, MN 55455.}}
    \date{}
    \maketitle
} \fi
\if1\blind
{
    \title{\Title}
    \author{}
    \date{}
    \maketitle
} \fi

\begin{abstract}\footnotesize 
 This article proposes a novel causal discovery and inference method called GrIVET for a Gaussian directed acyclic graph with unmeasured confounders. GrIVET consists of an order-based causal discovery method and a likelihood-based inferential procedure. For causal discovery, we generalize the existing peeling algorithm to estimate the ancestral relations and candidate instruments in the presence of hidden confounders. Based on this, we propose a new procedure for instrumental variable estimation of each direct effect by separating it from any mediation effects. For inference, we develop a new likelihood ratio test of multiple causal effects that is able to account for the unmeasured confounders. Theoretically, we prove that the proposed method has desirable guarantees, including robustness to invalid instruments and uncertain interventions, estimation consistency, low-order polynomial time complexity, and validity of asymptotic inference. Numerically, GrIVET performs well and compares favorably against state-of-the-art competitors. Furthermore, we demonstrate the utility and effectiveness of the proposed method through an application inferring regulatory pathways from Alzheimer's disease gene expression data.
\end{abstract}

\noindent
Keywords: Causal discovery, Gaussian directed acyclic graph, Invalid instrumental variables, Uncertain interventions, Simultaneous inference, Gene regulatory network.

\spacingset{1} 

\section{Introduction}

Understanding causal relations is part of the foundation of intelligence. A directed acyclic graph (DAG) is often used to describe the causal relations among multiple interacting units \citep{pearl2009causality}. Unlike classical causal inference tasks where the DAG is determined a priori, causal discovery aims to learn a graphical representation from data. It is useful for forming data-driven conjectures about the underlying mechanism of a complex system, including gene networks \citep{sachs2005causal}, functional brain networks \citep{liu2017effective}, manufacturing pipelines \citep{kertel2022learning}, and dynamical systems \citep{li2020causal}. In such a situation, randomized experiments are usually unethical or infeasible, and unmeasured confounders commonly arise in practice. The presence of latent confounders can bias the causal effect estimation and even distort causal directions, making causal discovery challenging. To treat latent confounders, we use additive interventions as instrumental variables (IVs), which are well-developed in conventional causal inference \citep{angrist1996identification} yet are less explored in causal discovery of a large-scale network. In this article, we focus on a Gaussian DAG model with hidden confounders and develop methods that integrate the discovery and inference of causal relations within the framework of uncertain additive interventions (the targets of interventions are unknown).

Causal discovery has been extensively studied \citep{zheng2018dags,aragam2019globally,gu2019penalized,lee2022functional,zhao2022learning,li2023nonlinear}; see \cite{drton2017structure,heinze2018causal,glymour2019review,vowels2021d} for comprehensive reviews. For observational data (without external interventions), some methods are able to treat hidden confounding by either (a) producing less informative discoveries, like a partial ancestral graph \citep{colombo2012learning} rather than a DAG, or (b) employing a certain deconfounding strategy \citep{frot2019robust,shah2020right} based on the pervasive confounding assumption. However, the former may not reveal essential information, such as causal directions, while the latter can be inconsistent in low-dimensional situations and may not necessarily outperform the naive regression \citep{grimmer2020naive}.
Thus, external interventions are useful to provide more information about causal relations while relaxing the requirements on latent confounding.

As an example of external (additive) interventions, IVs have been well developed in conventional causal inference to tackle unmeasured confounding; see \cite{lousdal2018introduction} for a survey. In a classical bivariate setting where the causal direction is known, an IV is required to influence the response variable only through the cause variable, which is often fragile in practice \citep{murray2006avoiding}. For instance, genetic variants like single nucleotide polymorphisms (SNPs) are used as IVs in Mendelian randomization (MR) analysis to discover putative causal genes of complex traits, where the IV conditions are commonly violated due to the (horizontal) pleiotropy. Remedying these invalid IVs has been the subject of recent work in causal inference \citep{kang2016instrumental,guo2018confidence,windmeijer2019use,burgess2020robust}. The discussion of IV estimation in graphical modeling, however, remains limited. The methods of \cite{oates2016estimating,chen2018two} estimate the graph given valid IVs, while the work of \cite{li2023inference} propose the peeling algorithm to construct the DAG in the case of uncertain interventions and invalid IVs. None of these methods permit latent confounding. A recent work \citep{xue2020inferring} discusses causal discovery of a bivariate mixed effect graph where confounders and invalid IVs are allowed, but it remains unclear how to extend it to a large-scale causal network.

Moreover, despite the progress in causal discovery, inference about the discovered relations is often regarded as a separate task and has received less attention in the literature. Notable exceptions include recent advances in graphical modeling \citep{jankova2018inference,li2020likelihood,shi2023testing,wang2023confidence} and mediation analysis \citep{chakrabortty2018inference,shi2021testing,li2022sequential}; however, these methods cannot account for latent confounders. Indeed, due to unmeasured confounding, the probability distribution of observed variables is no longer locally Markovian with respect to the DAG \citep{pearl2009causality}, rendering  these approaches inappropriate. Consequently, there is a pressing need for new inference methodologies.

This article contributes to the following aspects.
\begin{itemize}
 \item For modeling, we establish the identifiability conditions for a Gaussian DAG with latent confounders utilizing additive interventions. To our knowledge, this result is the first of its kind. Importantly, the conditions allow the interventions to have unknown and multiple targets, which is suitable for multivariate causal analysis \citep{murray2006avoiding}. 

 \item For methodology, we develop a novel method named the Graphical Instrumental Variable Estimation and Testing (GrIVET), integrating order-based causal discovery and likelihood-based inference. For causal discovery, we estimate the ancestral relations and candidate IVs with a modified peeling algorithm to treat unmeasured confounding. On this basis, we propose a sequential procedure to estimate each direct effect using IVs, where a working response regression is used to separate the direct effect from the mediation effects. Regarding inference, we develop a new likelihood ratio test of multiple causal effects to account for unmeasured confounders. 

  \item For theory, we show that GrIVET enjoys desired guarantees. 
  In particular, it consistently estimates the DAG structure and causal effects even when some interventions do not meet the IV criteria. As for computation, only $O((p+|\mathcal E^{+}|)\times \log(s) \times (q^3 + nq^2))$ operations are required almost surely, where $p$ and $q$ are the numbers of primary and intervention variables, $s$ is sparsity, $|\mathcal E^+|$ is the size of the ancestral relation set, and $n$ is the sample size. Moreover, under the null hypothesis, we establish the convergence of the likelihood ratio statistic to the null distribution in high-dimensional situations, ensuring the validity of asymptotic inference.
  
  \item The simulation studies and an application to the Alzheimer's Disease Neuroimaging Initiative dataset demonstrate the utility and effectiveness of the proposed methods. The implementation of GrIVET is available at \url{https://github.com/chunlinli/grivet}.

\end{itemize}

The rest of the article is structured as follows.
Section \ref{section:model} introduces a linear structural equation model with hidden confounders and establishes its identifiability.
Section \ref{section:discovery} presents a novel order-based method for causal discovery and effect estimation.
Section \ref{section:inference} develops a likelihood ratio test for simultaneous inference of causal effects.
Section \ref{section:causal-discovery-theory} provides theoretical justification of the proposed method.
Section \ref{section:numerical} performs simulation studies, followed by an application to infer gene pathways with gene expression and SNP data. 
Finally, Section \ref{section:discussion} concludes the article.
The Appendix contains supporting lemmas, while the Supplementary Materials include illustrative examples, technical proofs, and additional simulations.

\section{Causal graphical model with confounders}\label{section:model}

\subsection{Structural equations with confounders}\label{section:sem-dag}

We consider a structural equation model with $p$ primary variables $\bm Y = (Y_1,\ldots,Y_p)^\top$ 
and $q$ intervention variables $\bm X=(X_1,\ldots, X_q)^\top$,
\begin{equation}\label{equation:model}
  \begin{split}
  \bm Y = \mathbf{U}^\top \bm Y + \mathbf{W}^\top \bm X + \bm \varepsilon,
  \quad \bm\varepsilon \sim N(\bm 0,\bm\Sigma),
  \quad \Cov(\bm \varepsilon, \bm X) = \bm 0,
  \end{split}
\end{equation}
where $\mathbf{U}_{p\times p}$ is a matrix describing the causal influences among $\bm Y$, $\mathbf{W}_{q\times p}$ is a matrix representing the interventional effects of $\bm X$ on $\bm Y$, and $\bm \varepsilon$ is a vector of possibly correlated errors.
Specifically, 
\begin{itemize}
    \item The parameter matrix $\mathbf{U}$, which is of primary interest, has a causal interpretation in that $\mathrm{U}_{kj}\neq 0$ indicates that $Y_k$ is a cause of $Y_j$, denoted by $Y_k\to Y_j$. Thus, $\mathbf U$ represents a directed graph among primary variables. In what follows, we will focus on a directed acyclic graph (DAG), where no directed cycle is permissible and $\mathbf U$ is subject to the acyclicity constraint \citep{zheng2018dags,yuan2019constrained}. 
    
    \item The intervention variables $\bm X$ and errors $\bm\varepsilon$ are uncorrelated by reparameterization. As a result, $\mathbf W$ is associational instead of causal. Here, $\mathrm{W}_{lj}\neq 0$ indicates that $X_l$ intervenes on $Y_j$, denoted by $X_l\to Y_j$. As $\bm X$ represents external interventions, no directed edge from a primary variable to an intervention variable is allowed.
    
    \item A non-diagonal $\bm\Sigma$ indicates the presence of unmeasured confounders. For instance, $\bm\varepsilon = \bm\Phi^\top\bm\eta + \bm e$ can be (not uniquely) written as a sum of correlated components $\bm\Phi^\top\bm\eta$ and independent components $\bm e$ so that $\bm\Sigma = \bm\Phi^\top\bm\Phi + \Diag(\sigma_1^2,\ldots,\sigma_p^2)$, where $\bm\Phi_{r\times p}$ is the matrix of confounding effects, $\bm\eta\sim N(\bm 0,\mathbf I_{r\times r})$ represents $r$ independent confounding sources, and $\bm e\sim N(\bm 0,\Diag(\sigma_1^2,\ldots,\sigma_p^2))$ represents $p$ independent errors. Whenever $\Sigma_{jk}\neq 0$ for some distinct $(j,k)$, we have $\Sigma_{jk}=\sum_{m=1}^r \Phi_{mj}\Phi_{mk}\neq 0$, implying that some confounding variable $\eta_m$ influences both $Y_j$ and $Y_k$.
    
\end{itemize}
As such, $(\mathbf U,\mathbf W)$ together represents a directed graph of $p$ primary variables and $q$ intervention variables, denoted as $\mathcal G = (\bm X,\bm Y; \mathcal E,\mathcal I)$, where $\mathcal E = \{(k,j): \mathrm{U}_{kj}\neq 0\}$ is the set of primary variable edges and $\mathcal I = \{(l,j): \mathrm W_{lj}\neq 0\}$ is the set of intervention edges. 
In $\mathcal G$, (a) if $Y_k\to Y_j$, then $Y_k$ is a parent of $Y_j$, and $Y_j$ is a child of $Y_k$, (b) if $Y_k\to \cdots \to Y_j$ (a directed path from $Y_k$ to $Y_j$), then $Y_k$ is an ancestor of $Y_j$, and $Y_j$ is a descendant of $Y_k$, and (c) if $Y_k \to \cdots \to Y_m \to \cdots \to Y_j$, then $Y_m$ is a mediator of $Y_k$ and $Y_j$.
In what follows, for a graph $\mathcal G$, denote the parent set of $Y_j$ as $\pa_{\mathcal G}(j) = \{ k : Y_k\to Y_j \}$, the ancestor set of $Y_j$ as $\an_{\mathcal G}(j) = \{ k : Y_k\to\cdots\to Y_j \}$, and the intervention set of $Y_j$ as $\inter_{\mathcal G}(j) = \{l : X_l\to Y_j\}$.
For $(k,j)$ such that $Y_k\to\cdots \to Y_j$, denote the mediator set as $\me_{\mathcal G}(k,j) = \{ m : Y_k\to\cdots\to Y_m\to\cdots\to Y_j\}$.

\subsection{Identifiability and instrumental variables}\label{section:iv-identifiability}

The causal parameter matrix $\mathbf U$ is generally non-identifiable\footnote{{The causal parameter $\mathbf U$ is said to be identifiable if for any $(\mathbf U,\mathbf W,\bm\Sigma)$ and $(\mathbf U',\mathbf W',\bm\Sigma')$, we have $\mathbb P_{\mathbf U,\mathbf W,\bm\Sigma} = \mathbb P_{\mathbf U',\mathbf W',\bm\Sigma'}$ implies $\mathbf U = \mathbf U'$. Otherwise, it is said to be non-identifiable.}} without further conditions on the Gaussian errors $\bm\varepsilon$ or the interventions $\bm X$. 
Without invoking external interventions ($\mathbf W \equiv \bm 0$), $\mathbf U$ can be identified under a certain error-scale assumption \citep{peters2014identifiability,ghoshal2018learning,rajendran2021structure}, which is sensitive to variable scaling such as the common practice of standardizing variables \citep{reisach2021beware}. 
To overcome this limitation, interventions are introduced to identify the causal parameters.
With suitable interventions, $\mathbf U$ is identifiable if no confounder is present in the model ($\bm\Sigma$ is diagonal) \citep{oates2016estimating,chen2018two,li2023inference}.
In addition, it is worth mentioning that $\mathbf U$ can be estimated without intervention if the errors $\bm\varepsilon$ are non-Gaussian \citep{shimizu2006linear,zhao2022learning}; however, such methods are not applicable in the case of unmeasured confounding.

This subsection establishes the identifiability of \eqref{equation:model} in the presence of unmeasured confounders using uncertain additive interventions (the targets of interventions are unknown) as IVs. 
To proceed, we introduce the notion of IV for our purpose.

\begin{definition}\label{definition:valid-iv}
    An intervention variable $X_l$ is said to be a valid IV of $Y_k$ in $\mathcal G$ if \textbf{(IV1)} $X_l$ intervenes on $Y_k$, namely $\mathrm W_{lk}\neq 0$, and \textbf{(IV2)} $X_l$ does not intervene on any other primary variable $Y_{k'}$, namely $\mathrm W_{lk'} = 0$ for $k'\neq k$.
    Otherwise, $X_l$ is called an invalid IV.
    Denote the valid IV set of $Y_k$ as $\iv_{\mathcal{G}}(k) = \{l : X_l \rightarrow Y_k, X_l\not\to Y_{k'}, k'\neq k \}$.
\end{definition}

\begin{remark} \label{remark:iv-conditions} 
Consider a bivariate case where we are interested in the potential causal effect $Y_1\to Y_2$.
In causal inference literature \citep{angrist1996identification,kang2016instrumental}, a valid IV $X$ of $Y_1$ is required to satisfy that (a) $X$ is related to the $Y_1$, referred to as relevance, (b) $X$ has no directed edge to $Y_2$, called exclusion, and (c) $X$ is not related to unmeasured confounders, called unconfoundedness. 
In \eqref{equation:model}, (IV1) is indeed the relevance property, (IV2) generalizes the exclusion property for causal discovery, and the requirement $\Cov(\bm\varepsilon,\bm X)=\bm 0$ corresponds to the unconfoundedness.  
\end{remark}

To identify $\mathbf U$, two challenges emerge as the confounders arise. First, determining causal directions in the graph becomes more challenging. In \eqref{equation:model}, because of hidden confounding, the distribution $\mathbb P(\bm Y\mid \bm X)$ does not admit the causal Markov property \citep{pearl2009causality} according to $\mathcal G$, that is, $Y_j$ is not independent of its non-descendants given $(\bm Y_{\pa_{\mathcal G}(j)},\bm X)$.
As a result, the existing methods based on this property can learn wrong causal directions due to misspecification. To identify causal directions, we formalize the concept of unmediated parents to highlight the causal relations that are critical in identification. 

\begin{definition}
    A primary variable $Y_k$ is an unmediated parent of $Y_j$ in $\mathcal G$ if $Y_k\to Y_j$ and there is no other directed path from $Y_k$ to $Y_j$. In other words, $Y_k$ is an unmediated parent of $Y_j$ if no mediator is between $Y_k$ and $Y_j$.
\end{definition}

{Another challenge comes from uncertain interventions and invalid IVs. Assigning valid IVs for each primary variable can be difficult when the targets of interventions are unknown.} Thus, it may be effective to construct a set of candidate IVs (including invalid IVs)
for each primary variable, on which we estimate the causal parameters $\mathbf U$.
To this end, we define $p$ candidate IV sets, one for each primary variable. 

\begin{definition}\label{definition:candidate-iv}
    An intervention variable $X_l$ is said to be a candidate IV of $Y_k$ in $\mathcal{G}$ if \textbf{(IV1')} $X_l$ intervenes on $Y_k$, and \textbf{(IV2')} $X_l$ does not intervene on any non-descendant of $Y_k$.
    Denote the candidate IV set of $Y_k$ by $\ca_{\mathcal G}(k) = \{l : X_l \rightarrow Y_k, X_l \rightarrow Y_j \text{ only if } k \in \an_{\mathcal{G}}(j)\}$. 
\end{definition}

The candidate IVs of $Y_k$ include all valid IVs of $Y_k$, but not vice versa. A candidate IV of $Y_k$ may be invalid, as it could intervene on descendants of $Y_k$.

\begin{theorem}[Identifiability] \label{theorem:identifiability} 
Suppose 
\begin{enumerate}
  \item [(A1)] $\Cov(\bm X)$ is positive definite.
  
  \item [(A2)] $\Cov(Y_j,  X_l \mid  \bm X_{\{1,\ldots,q\}\setminus\{l\}})\neq 0$ whenever $X_l$ intervenes on an unmediated parent of $Y_j$.
  
  \item [(A3)] \emph{(Majority rule)} $|\iv_{\mathcal G}(k)|>|\ca_{\mathcal G}(k)|/2$; $k=1,\ldots,p$.
\end{enumerate}
Then $(\mathbf{U},\mathbf{W},\bm \Sigma)$ in \eqref{equation:model} are identifiable in that if $(\mathbf U, \mathbf W, \bm\Sigma)$ and $(\mathbf U', \mathbf W', \bm\Sigma')$ encode the same probability distribution, then $(\mathbf U, \mathbf W, \bm\Sigma) = (\mathbf U', \mathbf W', \bm\Sigma')$. 
\end{theorem}

To our knowledge, Theorem \ref{theorem:identifiability} is a new result for Gaussian DAG with hidden confounding, establishing the identifiability of all parameters in \eqref{equation:model}. In fact, if the causal parameter $\mathbf U$ is identifiable, then so are parameters $\mathbf W,\bm\Sigma$.
Regarding the conditions, (A1) states that $\Cov(\bm X)$ has full rank, which is common in the IV literature \citep{kang2016instrumental,chen2018two}. Note that (A1) permits discrete IV variables such as SNPs in data analysis. 
(A2) requires the interventional effects through unmediated parents not to cancel out when an invalid IV has multiple targets. (A3) requires valid IVs to dominate invalid ones so that the causal effect can be identified in the presence of latent confounders. Such a condition has been used in the causal inference literature \citep{kang2016instrumental,windmeijer2019use}. As shown in Supplementary Materials Section 1, when (A3) fails, \eqref{equation:model} can be non-identifiable. By comparison, (A1)--(A2) together with (A4) are used for model identification in the absence of unmeasured confounding \citep{li2023inference}.
\begin{itemize}
    \item [(A4)] Each $Y_k$ is intervened by at least one valid IV.
\end{itemize}
Noting that (A4) is implied by (A3), treating hidden confounding demands stronger conditions in view of Theorem \ref{theorem:identifiability}.

\section{Causal discovery}\label{section:discovery}

This section proposes a novel IV method to learn a DAG with unmeasured confounders. 
First, we introduce the ancestral relation graph (ARG), which, together with the candidate IV sets in Section \ref{section:iv-identifiability}, constitutes a basis for the proposed method. 

\begin{definition}[Ancestral relation graph]\label{definition:ancestral-relation-graph}
For a DAG $\mathcal G = (\bm X,\bm Y; \mathcal E, \mathcal I)$, its ancestral relation graph is defined as $\mathcal G^{+} = (\bm X,\bm Y; \mathcal E^{+}, \mathcal I^{+})$, where 
\begin{equation*}
    \mathcal E^{+} = \Big\{ (k,j) : k\in\an_{\mathcal G}(j) \Big\}, \qquad
    \mathcal I^{+} = \Big\{ (l,j) : l \in \bigcup_{k\in\an_{\mathcal G}(j)\cup\{ j\}} \inter_{\mathcal G}(k) \Big\}.
\end{equation*}
\end{definition}

Here, $\mathcal G^{+}$ is a super-DAG of $\mathcal G$ in that $\mathcal E^{+}\supseteq \mathcal E$ is the set of ancestral relations, $\mathcal I^{+}\supseteq \mathcal I$ is a superset of interventional relations, and $\mathcal G^{+}$ is acyclic.
Note that $\mathcal E^{+}$ defines a partial order for the primary variables $\bm Y$ in that $Y_{k}\prec_{\mathcal G} Y_j$ whenever $(k,j)\in\mathcal E^{+}$.
Without confounding, $\mathbf U$ can be consistently estimated via direct regressions according to the known $\mathcal G^{+}$ \citep{shojaie2010penalized}, where $\mathcal G^{+}$ can be recovered by the peeling algorithm \citep{li2023inference}. However, this approach no longer applies in the presence of hidden confounders. 

To address this obstacle, Sections \ref{section:unknown-order-peeling}--\ref{section:unknown-order-peeling2} modify the peeling algorithm to construct the ARG $\mathcal G^{+}$ and the candidate IV sets $\{\ca_{\mathcal G}(k)\}_{1\leq k\leq p}$, and then Sections \ref{section:known-order-estimation}--\ref{section:finite-sample-estimation} develop a method to estimate $\mathbf U$ assuming the ARG and candidate IVs are known.

\subsection{Identification of \texorpdfstring{$\mathcal G^{+}$}{G+} and candidate IVs} \label{section:unknown-order-peeling}

In this subsection, we modify the peeling algorithm, originally designed for a model without unmeasured confounders \citep{li2023inference}, to uncover $\mathcal G^{+}$ and $\{\ca_{\mathcal G}(k)\}_{1\leq k\leq p}$ in the presence of hidden confounders, of which the results can be subsequently used as the inputs for identification of $\mathbf U$ in Section \ref{section:known-order-estimation}. 
The modified peeling algorithm essentially requires $p$ regressions to identify the ARG and candidate IVs, which is suited for large-scale causal discovery. Moreover, the produced ARG and candidate IV sets enjoy desirable statistical properties; see Section \ref{section:causal-discovery-theory}.

Let us begin with an observation that \eqref{equation:model} can be rewritten as 
\begin{equation}\label{equation:y-x-model}
\bm Y = \mathbf{V}^\top \bm X + (\mathbf{I}-\mathbf{U}^\top)^{-1} \bm \varepsilon,
\end{equation}
where $\mathbf{V} = \mathbf{W} (\mathbf{I} - \mathbf{U})^{-1}$ and $\mathrm{V}_{lj} = \sum_{k=1}^p \mathrm W_{lk} (\mathrm{I}_{kj} + \mathrm{U}_{kj} + \cdots + (\bs{U}^{p-1})_{kj})$.
Intuitively, $\mathrm V_{lj}\neq 0$ implies the dependence of $Y_j$ on $X_{l}$ through a directed path $X_l\to Y_{k} \to \cdots \to Y_j$, and hence that $X_l$ intervenes on $Y_j$ itself (when $k=j$) or its ancestor $Y_{k}$ (when $k\neq j$). 
In cases where $X_l$ intervenes exclusively on one primary variable, the following proposition provides insights into the connection between $\mathbf{V}$ and $\mathcal G^{+}$.

\begin{proposition} \label{proposition:peeling} 
   Suppose Assumptions (A1), (A2), and (A4) are satisfied. There exists at least one intervention variable $X_l$ such that $\mathrm V_{lk}\neq 0$ and $\mathrm V_{lk'}=0$ for $k'\neq k$ if and only if $Y_k$ is a leaf node (has no descendant). Moreover, such $X_l$ is a valid IV of $Y_k$ in $\mathcal G$.
\end{proposition}

Proposition \ref{proposition:peeling} suggests that the leaves and their valid IVs in $\mathcal G$ can be identified by 
\begin{equation}\label{equation:leaf}
\begin{split}
    \textsc{leaf}(\mathcal G) 
    & = \{k: \text{ for some $l$, } \mathrm V_{lk} \neq 0 \text{ and } \mathrm V_{lk'}=0 \text{ for all } k'\neq k \} \\
    & = \{ k :  k = \argmax_j |\mathrm V_{lj}| \text{ for some } l = \argmin_{\|\mathbf{V}_{l,+}\|_0 > 0} \|\mathbf{V}_{l,+}\|_0 \},\\
    \iv_{\mathcal G}(k) & = \{ l : \mathrm V_{lk}\neq 0 \text{ and } \mathrm V_{l k'} = 0 \text{ for all } k' \neq k \} \\
    & =  \{l:  l = \argmin_{\|\mathbf{V}_{l,+}\|_0 > 0} \|\mathbf{V}_{l,+}\|_0 \text{ and } k = \argmax_j |\mathrm V_{lj}| \}, \quad k\in\textsc{leaf}(\mathcal G).
\end{split}
\end{equation}
After the leaf nodes are learned, we can remove them to obtain a sub-DAG. If $X_l$ is a valid IV of a non-leaf $Y_k$ in $\mathcal{G}$, its validity for $Y_k$ is retained in the sub-DAG, implying (A4) continues to hold. Moreover, Assumptions (A1)--(A2) are naturally upheld in the sub-DAG. Hence, the requirements of Proposition \ref{proposition:peeling} are satisfied in the sub-DAG, whose leaf variables and their valid IVs can be learned in the same fashion. As a result, we can successively identify and remove (i.e., peel) the leaf nodes from the DAG and sub-DAGs. This yields a topological order of primary variables but does not recover $\mathcal G^{+}$.

Next, we investigate how $\mathbf V$ can be further used to recover $\mathcal G^+$ with $\{\ca_{\mathcal G}(k)\}_{1\leq k\leq p}$. 
Subsequently, we use $\mathcal G^- = (\bm X^-,\bm Y^-;\mathcal E^-,\mathcal I^-)$ to denote a generic sub-DAG produced by peeling, where $\bm Y^-$ are the primary variables in $\mathcal G^-$ and $\bm Y\setminus\bm Y^-$ are peeled ones, $\bm X^-$ are intervention variables on $\bm Y^-$, $\mathcal E^-$ is the set of causal relations among $\bm Y^-$, and $\mathcal I^-$ is the set of interventional relations between $\bm X^-$ and $\bm Y^-$. Then each variable in $\bm Y^-$ is a non-descendant of each in $\bm Y\setminus \bm Y^-$. Moreover, $\textsc{leaf}(\mathcal G^-)$ and $\{ \iv_{\mathcal G^-}(k) \}_{k\in \textsc{leaf}(\mathcal G^-)}$ are identified by \eqref{equation:leaf}.

\begin{proposition}\label{proposition:peeling2} 
    Suppose Assumptions (A1), (A2), and (A4) are satisfied. Let $Y_k$ be a leaf node in $\mathcal{G}^-$ and $Y_j$ be in $\bm{Y} \setminus \bm{Y}^-$. Then the following statements are true.
    \begin{enumerate}[(A)]
    \item If $\mathrm V_{lj} \neq 0$ for all $l \in \iv_{\mathcal{G}^-}(k)$,  we have $(k,j) \in \mathcal{E}^{+}$.
    \item If $Y_k$ is an unmediated parent of $Y_j$, then $\mathrm V_{lj} \neq 0$ for all $l \in \iv_{\mathcal{G}^-}(k)$.
    \end{enumerate}
\end{proposition}

Proposition \ref{proposition:peeling2} outlines a method for identifying edges in $\mathcal G^+$ from the leaf variables of $\mathcal{G}^-$ to the peeled variables $\bm Y\setminus \bm Y^-$ by 
\begin{equation}\label{equation:unmediate-parent-edge}
    \{(k,j): Y_k \in \textsc{leaf}(\mathcal G^-), \  Y_j \in \bm Y\setminus \bm Y^- \text{ and } \mathrm V_{lj} \neq 0 \text{ for all } l \in \iv_{\mathcal{G}^-}(k) \}.
\end{equation}%
Specifically, (A) shows that any identified edge must be present in $\mathcal G^{+}$, so no extra edges are identified. Meanwhile, (B) shows that every directed edge from an unmediated parent must be correctly discovered. Importantly, the collection of all such edges suffices to recover all ancestral relationships, which guarantees that no edge in $\mathcal E^{+}$ is overlooked. Upon the identification of $\mathcal G^+$, the candidate IV sets can be learned by
\begin{equation}\label{equation:candidate-iv}
   \ca_\mathcal{G}(k) = \{l: (l,k) \in {\mathcal{I}}^{+} \text{ and } (l,j) \in {\mathcal{I}}^{+}, k \neq j \text{ only if } (k,j) \in {\mathcal{E}}^{+}\}, \quad 1\leq k\leq p.
\end{equation}
Consequently, Propositions \ref{proposition:peeling}--\ref{proposition:peeling2} enable the recovery of $\mathcal G^{+}$ and $\{\ca_{\mathcal G}(k)\}_{1\leq k\leq p}$.

\subsection{Finite-sample estimation of \texorpdfstring{$\mathcal G^{+}$}{G+} and candidate IVs}
 \label{section:unknown-order-peeling2}

This subsection implements the modified peeling algorithm delineated in Section \ref{section:unknown-order-peeling} to estimate $\mathcal G^+$ and $\{\ca_{\mathcal G}(k)\}_{1\leq k\leq p}$.
To proceed, suppose data matrices $\mathbf{Y}_{p\times n} = (\mathbf{Y}_{+,1},\ldots,\mathbf{Y}_{+,n})$ and $\mathbf{X}_{q\times n} = (\mathbf{X}_{+,1},\ldots,\mathbf{X}_{+,n})$ are given, where $(\mathbf{Y}_{+,i},\mathbf{X}_{+,i})_{i=1}^n$ are sampled from \eqref{equation:model} independently.
We estimate $\mathbf V$ by $\widehat{\mathbf V} = (\widehat{\mathbf V}_{+,1},\ldots,\widehat{\mathbf V}_{+,p})$ with sparse regressions
\begin{equation}\label{equation:estimate-v}
    \widehat{\mathbf V}_{+,j} = \argmin_{\bm\beta} \ \sum_{i=1}^n ( \mathrm{Y}_{j,i} - \bm\beta^\top \mathbf X_{+,i} )^2 \quad \text{s.t.} \quad \|\bm\beta\|_0 \leq \kappa'_j
\end{equation}
where $1\leq\kappa'_j\leq q$ is tuned by BIC for $1\leq j \leq p$\;
Moreover, the truncated Lasso penalty (TLP) \citep{shen2012likelihood} is used as the computational surrogate for $\|\cdot\|_0$, where TLP is defined as 
$\operatorname{TLP}_{\tau}(\bm\beta) = \sum_{j=1}^{r} \min(|\beta_j|/\tau,1)$ for $\bm\beta=(\beta_1,\ldots,\beta_r)$, and $\tau > 0$ is a hyperparameter in TLP; see Supplementary Materials Section 2 for details.
The modified peeling algorithm based on Section \ref{section:unknown-order-peeling} is summarized in Algorithm \ref{algorithm:peeling}.\footnote{{In Algorithm \ref{algorithm:peeling} Step 7, the indices of $\mathbf{V}$ are kept so that $\mathrm V_{lj}$ always represents the effect from $X_l$ to $Y_j$.}}

\begin{algorithm}[ht]
  \caption{Estimation of $\mathcal G^+$ and $\{\ca_{\mathcal G}(k)\}_{1\leq k\leq p}$} \label{algorithm:peeling}
  \KwIn{Data $\mathbf Y_{p\times n}$ and $\mathbf X_{q\times n}$;}

  Compute $\widehat{\mathbf V}$ via \eqref{equation:estimate-v}\;

  Initialize $\mathbf V\leftarrow \widehat{\mathbf V}$, $\widehat{\mathcal E}^+ \leftarrow \emptyset$, $\widehat{\mathcal I}^+\leftarrow \{(l,k): \widehat{\mathrm V}_{lk}\neq 0 \}$\; 
  
  Initialize $\mathcal G^-$ by $\bm Y^-\leftarrow \bm Y$, $\bm X^-\leftarrow\bm X$, $\mathcal E^-\leftarrow \widehat{\mathcal E}^+$, $\mathcal I^-\leftarrow \widehat{\mathcal I}^+$\;

  \While{${\bm Y}^-$ is not empty}{
    
    {Update $\textsc{leaf}(\mathcal G^-)$ and $\{\iv_{\mathcal G^-}(k)\}_{k\in\textsc{leaf}(\mathcal G^-)}$ via \eqref{equation:leaf}\;

    }

    Update $\widehat{\mathcal{E}}^{+}$ by adding \eqref{equation:unmediate-parent-edge}\;
    
    Update $\mathcal{G}^-$ by removing $\textsc{leaf}(\mathcal G^-)$ and ${\mathbf V}$ by keeping the columns in $\bm Y^-$\;  
  }
  
  Update $\widehat{\mathcal{E}}^{+}\leftarrow \{(k,j): Y_k \to \cdots \to Y_j \text{ in }  \widehat{\mathcal{E}}^{+} \}$\;
  Update $\widehat{\mathcal{I}}^{+}\leftarrow \{(l,j): (l,k) \in \widehat{\mathcal{I}}^{+} \text{ and } (k,j) \in \widehat{\mathcal{E}}^{+} \}$\;
  Update $\widehat{\ca}_\mathcal{G}(k)$ by \eqref{equation:candidate-iv}\; 

  \KwRet{$\widehat{\mathcal E}^+$, $\widehat{\mathcal I}^+$, and $\{\widehat{\ca}_{\mathcal G}(k)\}_{1\leq k\leq p}$}\;
\end{algorithm}

\subsection{Identification of \texorpdfstring{$\mathbf U$}{U}}\label{section:known-order-estimation}

In this subsection, we present a new method for identifying causal effects $\mathbf U$, using the ARG $\mathcal G^{+}$ and candidate IV sets $\{\ca_{\mathcal G}(k)\}_{1\leq k\leq p}$ as inputs. 
{Note that $\{\an_{\mathcal G}(k)\}_{1\leq k\leq p}$ and $\{\me_{\mathcal G}(k,j)\}_{(k,j)\in \mathcal E^+}$} can be derived from $\mathcal G^+$. 
Throughout this subsection, the subscript $\mathcal G$ is dropped for brevity and $\bm\alpha, \bm\beta, \bm\gamma$ denote nuisance parameters in regression.
Moreover, we assume that $\bm\varepsilon$ and $\bm X$ are independent to simplify the derivation; see Lemmas \ref{lemma:linear}--\ref{lemma:linear2} in the Appendix for the case with $\bm\varepsilon$ and $\bm X$ being uncorrelated.

\begin{figure}
    \centering
    \includegraphics[width=.8\textwidth]{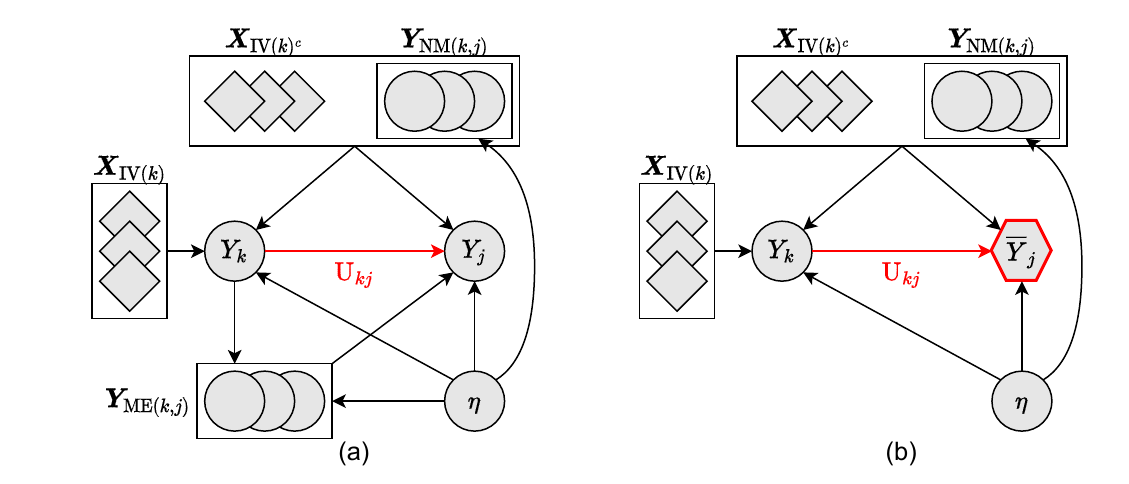}
    \caption{Estimation of causal parameter $\mathrm U_{kj}$. (a) Display of the relations among relevant variables. (b) Display of working response regression.}
    \label{fig:method}
\end{figure}

\paragraph{The case with all IVs being valid.}

We begin with a special case of \eqref{equation:model} where all IVs are valid, that is, $\ca(k) = \iv(k)$; $k=1,\ldots,p$.

{To estimate $\mathbf U$, note that $\mathbf U$ is supported on $\mathcal E^+$, namely $\mathbf U=(\mathbf U_{\mathcal E^+},\bm 0)$. Here, we consider estimating $\mathrm{U}_{kj}$, as well as selecting nonzero $\mathrm U_{kj}$ for graph recovery, for each $(k,j)\in\mathcal E^{+}$, as described in Figure \ref{fig:method} (a).}

{To pinpoint the difficulties and motivate our approach, we make the following observations.} 
First, regression of $Y_j$ on $Y_k$ together with covariates $(\bm Y_{\an(j)\setminus \{k\} },\bm X)$ can bias the estimation due to confounder $\eta$. Second, in hope of treating confounders one might replace $Y_k$ with its surrogate $\mathbb E(Y_k \mid \bm Y_{\an(k)}, \bm X)$ to regress $Y_j$ on $\mathbb E(Y_k \mid \bm Y_{\an(k)}, \bm X)$ with $(\bm Y_{\an(j)\setminus \{k\}}, \bm X_{\iv(k)^c})$ being covariates. 
However, this is also problematic. 
For explanation, note that $\an(j)\setminus \{k\}$ can be partitioned into mediators $\me(k,j)$ and non-mediators 
\begin{equation*}
    \nui(k,j) = \an(j)\setminus (\me(k,j)\cup \{k\}).
\end{equation*}
In Figure \ref{fig:method} (a), $\bm X_{\iv(k)}$ can be associated with $\eta$ given $\bm Y_{\an(j)\setminus \{k\} } = (\bm Y_{\me(k,j)},\bm Y_{\nui(k,j)})$, violating the unconfoundedness of IVs (Remark \ref{remark:iv-conditions}) and causing an estimation bias.
This is because the mediators $\bm Y_{\me(k,j)}$ generate additional associations after conditioning on them; see the Appendix for technical discussion using the concept of d-separation \citep{pearl2009causality}.

Now, we propose a new method, which eliminates the impact of mediators $\bm Y_{\me(k,j)}$ by introducing the working response $\overline{Y}_j = Y_j - {\mathbf{U}}_{\me(k,j),j}^\top \bm{Y}_{\me(k,j)}$, as depicted in Figure \ref{fig:method} (b). Of note, the definition of $\overline{Y}_j$ depends on $(k,j)$, which is dropped for simplicity. 
{As in  \citet{angrist1996identification}, we have} 
\begin{equation}\label{equation:iv-regression}
    \begin{split}
        & \E \Big(\overline{Y}_j
        \mid \bm Y_{\nui(k,j)}, \bm X\Big) \\
        \overset{\text{(i)}}{=}\ & \mathrm U_{kj} \E\Big(Y_k \mid \bm Y_{\nui(k,j)}, \bm X\Big)
         + \sum_{k'\in\nui(k,j)} \mathrm U_{k'j} Y_{k'} + 
         \sum_{l\notin \iv(k)}
         \mathrm W_{lj} X_{l} + \E \Big(\varepsilon_j \mid \bm Y_{\nui(k,j)}, \bm X \Big) \\
        \overset{\text{(ii)}}{=}\ & \mathrm U_{kj}\widetilde{Y}_k + \bm \gamma^\top \bm Z,
    \end{split}
\end{equation}
where $\widetilde{Y}_k = \mathbb E(Y_k \mid \bm Y_{\nui(k,j)}, \bm X)$, $\bm Z = (\bm Y_{\nui(k,j)}, \bm X_{\ca(k)^c}) = (\bm Y_{\nui(k,j)}, \bm X_{\iv(k)^c})$, equality (i) follows from \eqref{equation:model}, and equality (ii) holds because $\mathbb E \left(\varepsilon_j \mid \bm Y_{\nui(k,j)}, \bm X \right)$ is a linear combination of $(\bm Y_{\nui(k,j)}, \bm X_{\iv(k)^c})$ by Lemma \ref{lemma:linear} in Appendix. Observe that $\widetilde{Y}_k$ depends on $\bm X_{\iv(k)}$ while $\bm Z$ does not.
As a result, the $\mathrm U_{kj}$ is identified through the working response regression.

This approach requires the knowledge of $\mathbf{U}_{\me(k,j),j}$ prior to identifying $\mathrm U_{kj}$. Given $\mathcal G^{+}$, we develop a sequential procedure to learn $\bs U$. First, we identify $\mathrm{U}_{kj}$ for each pair $(k,j)$ such that the longest path in $\mathcal G^{+}$ between $k$ and $j$ is equal to $d=1$. Then for $(k,j)$ such that the longest path in $\mathcal G^{+}$ between $k$ and $j$ is $d=2$, the effects of mediators $\mathbf U_{\me(k,j),j}$ are available. Thus, we can identify $\mathrm{U}_{kj}$ in \eqref{equation:iv-regression}. Proceed similarly for $d = 3,4,5,\ldots$ until all pairs in $\mathcal E^{+}$ have been identified.

\paragraph{The case with invalid IVs.}

In general, $\ca(k)\supseteq \iv(k)$ because of invalid IVs, where $\ca(k)$ is known but $\iv(k)$ is unknown.
{Similar to \citet{kang2016instrumental}, we have} 
\begin{equation}\label{equation:invalid-iv-regression}
    \begin{split}
        &\E \Big(\overline{Y}_j
        \mid  \bm Y_{\nui(k,j)}, \bm X \Big) \\
        {=}\ & \mathrm U_{kj} \E\Big(Y_k \mid \bm Y_{\nui(k,j)}, \bm X\Big)
         + 
         \sum_{k'\in\nui(k,j)} 
         \mathrm U_{k'j} Y_{k'} + 
         \sum_{l\notin \iv(k)}
         \mathrm W_{lj} X_{l} + \E \Big(\varepsilon_j \mid \bm Y_{\nui(k,j)}, \bm X \Big) \\
        \overset{\text{(iii)}}{=}\ & \mathrm U_{kj}\widetilde{Y}_k + \bm \gamma^\top \bm Z + 
        \sum_{l\in\ca(k)\setminus \iv(k)}
        \beta_l X_{l},
    \end{split}
\end{equation}
where $\widetilde{Y}_k = \mathbb E(Y_k \mid \bm Y_{\nui(k,j)}, \bm X)$, $\bm Z = (\bm Y_{\nui(k,j)}, \bm X_{\ca(k)^c})$, equality (iii) holds by Lemma \ref{lemma:linear} in Appendix, and $\beta_l =\mathrm W_{lj} \neq 0$ indicates $X_l$ is an invalid IV for $Y_k$. However, since $\iv(k)$ has not been identified and $\widetilde{Y}_k$ depends on $\bm X_{\ca(k)}$, the representation of (iii) may not be unique. When the majority rule (A3) is satisfied by the DAG, the term (iii) admits the unique expression as in \eqref{equation:invalid-iv-regression}, providing the identification of $\mathrm U_{kj}$.
This leads to a sparse regression for an infinite sample
\begin{equation}\label{equation:identification-u}
    \min_{\mathrm{U}_{kj},\bm\beta,\bm\gamma} \ \mathbb E\left(\overline{Y}_j - \mathrm{U}_{kj} \widetilde{Y}_k - \bm\gamma^\top\bm Z - \bm\beta^\top \bm X_{\ca(k)} \right)^2 \quad \text{s.t.} \quad \|\bm\beta\|_0\leq \kappa,
\end{equation}
where $0\leq\kappa <|\ca(k)|/2$ is an integer-valued hyperparameter controlling the sparsity of $\bm\beta$.

\subsection{Finite-sample estimation of U} \label{section:finite-sample-estimation}

Suppose $(\mathbf Y_{p\times n},\mathbf X_{q\times n})$ are given. To estimate $\mathrm{U}_{kj}$, noting that $\widetilde{Y}_k$ is linear in $(\bm Y_{\nui(k,j)},\bm X)$ by Lemma \ref{lemma:linear}, we estimate $\widetilde{\mathrm{Y}}_{k,i}$ by $\widehat{\mathrm{Y}}_{k,i} = \widehat{\bm \alpha}_1^\top \bs X_{+,i} + \widehat{\bm \alpha}_2^\top \bs Y_{\nui(k),i}$, where $(\widehat{\bm \alpha}_1,\widehat{\bm \alpha}_2)$ solves
\begin{equation}\label{equation:imputation}
    \begin{split}
        \min_{\bs \alpha_1,\bs \alpha_2} \ \sum_{i=1}^n \Big( \mathrm{Y}_{k,i} - \bm \alpha_1^\top \bs X_{+,i} + \bm\alpha_2^\top \bs Y_{\nui(k),i} \Big)^2 \quad 
        \text{s.t.} \quad 
        \|\bs\alpha_1\|_0 + \|\bs \alpha_2\|_0 \leq \nu_1,
    \end{split}
\end{equation}
with $\nu_1$ being a tuning parameter.
Let the final estimate $\widehat{\mathrm{U}}_{kj}$ with $(\widehat{\bm\beta},\widehat{\bm\gamma})$ be the solution to the working response regression (provided that $\widehat{\bs U}_{\me(k,j),j}$ are available)
\begin{equation}\label{equation:working-response-regression}
    \begin{split}
        \min_{\mathrm U_{kj},\bm\beta,\bm\gamma} \ 
        \sum_{i=1}^n \Big( \Big( \mathrm{Y}_{j,i} - \widehat{\bs U}^\top_{\me(k,j),j} \bs Y_{\me(k,j),i}\Big) - \mathrm U_{kj} \widehat{\mathrm{Y}}_{k,i} - \bm\beta^\top\bm X_{\ca(k),i} - \bm\gamma^\top\bm Z_i \Big)^2 \\
        \text{s.t.} \quad \rho(\mathrm U_{kj}) + \|\bm\beta\|_0\leq \kappa, \quad  \|\bm\gamma\|_0\leq \nu_2,
    \end{split}
\end{equation}
where $0\leq\kappa \leq |\ca(k)|/2$ and $0\leq \nu_2\leq |\nui(k,j)|+|\ca(k)^c|$ are tuning parameters. {Depending on the purpose, $\rho(\cdot) = \I(\cdot\neq 0)$ for graph recovery and $\rho(\cdot) = 0$ for effect estimation without selection.}
In \eqref{equation:imputation}--\eqref{equation:working-response-regression}, 
$\nu_1,\nu_2$ are added to treat possible high-dimensional situations and the hyperparameters are tuned by BIC.
Algorithm \ref{algorithm:pruning} summarizes the procedure. 

\begin{algorithm}[ht]
  \caption{Estimation of $\mathbf U$} \label{algorithm:pruning}
  \KwIn{Data $\mathbf Y_{p\times n}$ and $\mathbf X_{q\times n}$, ARG $\mathcal G^+$ and candidate IV sets $\{\ca_{\mathcal G}(k)\}_{1\leq k\leq p}$; }

  Initialize $\widehat{\mathbf U}\leftarrow \bm 0$ and $d\leftarrow 1$\;
  \While{$d \leq $ the length of the longest directed path in $\mathcal G^+$}{
        For $(k,j)\in \mathcal E^+$ so that the length of the longest directed path from $Y_k$ to $Y_j$ is $d$, 
        estimate $\widehat{\mathrm U}_{kj}$ with \eqref{equation:imputation}--\eqref{equation:working-response-regression}\;
        Update $d\leftarrow d+1$\;
  }
  \KwRet{$\widehat{\mathbf U}$}\;
\end{algorithm}

\section{Likelihood inference}\label{section:inference}

This section develops a likelihood ratio test for the presence of multiple directed edges. 
Let $\mathcal{H} \subseteq\{(k, j): k \neq j, \ 1\leq k, j\leq p\}$ be a hypothesized edge set for primary variables $\bm Y$, where $(k,j) \in \mathcal{H}$ specifies a (hypothesized) directed edge $Y_k \rightarrow Y_j$ in \eqref{equation:model}. 
Now consider simultaneous testing of directed edges,
\begin{equation}\label{equation:test}
H_{0}: \mathrm{U}_{k j}=0 \text { for all }(k, j) \in \mathcal{H} \quad \text { versus } \quad H_{a}: \mathrm{U}_{k j} \neq 0 \text { for some }(k, j) \in \mathcal{H}.
\end{equation}
The null hypothesis $H_0$ asserts that all hypothesized edges in $\mathcal H$ are absent in the true graph $\mathcal G$. Rejecting $H_0$ indicates that at least one hypothesized edge in $\mathcal H$ presents in $\mathcal G$.

\paragraph*{The likelihood ratio.}

Given $\mathcal G^{+} = (\bm X,\bm Y; \mathcal E^{+},\mathcal I^{+})$, let $\bm\theta(\mathcal G^{+}) = (\mathbf U,\mathbf W)$ encode the coefficient parameters in $\mathcal G^{+}$, where $\mathbf U = (\mathbf U_{\mathcal E^{+}},\bm 0)$ and $\mathbf W = (\mathbf W_{\mathcal I^{+}},\bm 0)$. 
As such, the adjacency matrix $\mathbf U$ automatically meets the acyclicity constraint.
Given a random sample ${(\mathbf{Y}_{+,i}, \mathbf{X}_{+,i})}_{i=1}^{n}$, the log-likelihood is written as (up to an additive constant)
\begin{equation}\label{equation:log-likelihood}
L(\bm{\theta}(\mathcal G^{+}), \bm \Omega)=-\frac{1}{2} \sum_{i=1}^{n}\left\| \bm{\Omega}^{1/2} \left(\left(\mathbf{I}-\mathbf{U}^{\top}\right) \mathbf{Y}_{+,i}-\mathbf{W}^{\top} \mathbf{X}_{+,i}\right) \right\|_2^2 + \frac{n}{2} \log \det(\bm{\Omega}),
\end{equation}
where $\bm\Omega = \bm\Sigma^{-1}$ is the inverse of $\bm\Sigma$ in \eqref{equation:model}.
Then the maximum likelihood estimation (MLE) of \eqref{equation:model} can be written as 
\begin{equation}\label{equation:MLE}
    \max_{(\mathcal G^{+},\bm\Omega)} \max_{\bm\theta(\mathcal G^{+})} L(\bm{\theta}(\mathcal G^{+}), \bm \Omega).
\end{equation}
In view of \eqref{equation:MLE}, to obtain a likelihood ratio statistic for \eqref{equation:test} we need to compute the following quantities: (1) a consistent estimate $\widehat{\mathcal G}^{+}$ of $\mathcal G^{+}$, (2) a consistent estimate $\widehat{\bm\Omega}$ of $\bm\Omega$, and (3) two estimates, $\widehat{\bm\theta}^{(0)}$ and $\widehat{\bm\theta}^{(1)}$, of $\bm\theta(\mathcal G^{+})$ under $H_0$ and $H_a$, respectively.
This leads to the likelihood ratio defined as 
\begin{equation}\label{equation:likelihood-ratio}
L (\widehat{\bm{\theta}}^{(1)}, \widehat{\bm{\Omega}})- L (\widehat{\boldsymbol{\theta}}^{(0)}, \widehat{\bm{\Omega}}),
\end{equation}
where $\mathcal G^{+}$ is estimated by Algorithm \ref{algorithm:peeling} and ${\bm\Omega}$ is estimated from the residuals after fitting model \eqref{equation:model} via Algorithm \ref{algorithm:pruning}.

\paragraph*{Inference subject to acyclicity.}

In classical models, a likelihood ratio of form \eqref{equation:likelihood-ratio} has a nondegenerate and tractable limiting distribution, typically a chi-squared distribution with degrees of freedom $|\mathcal H|$. However, the likelihood ratio for \eqref{equation:test} may behave differently from classical ones since \eqref{equation:likelihood-ratio} may be degenerate or intractable, as to be explained.

First, note that the maximum likelihood subject to a wrong ARG $\widetilde{\mathcal G}^{+}\not\supseteq \mathcal G$ tends to be smaller than that subject to the correct $\mathcal G^{+}$, 
that is, 
\begin{equation*}
    \max_{\widetilde{\mathcal G}^{+}\not\supseteq \mathcal G}\max_{\bm\theta(\widetilde{\mathcal G}^{+}),\bm\Omega} L(\bm\theta(\widetilde{\mathcal G}^{+}),\bm\Omega) < \max_{\bm\theta({\mathcal G}^{+}),\bm\Omega} L(\bm\theta(\mathcal G^{+}),\bm\Omega),
\end{equation*}
as $n\to\infty$ under some regularity conditions for consistency.
Thus, we assume $\widehat{\mathcal G}^{+}=\mathcal G^{+}$ in this paragraph.
Then $\widehat{\bm\theta}^{(0)}$ is the MLE subject to ${\mathcal G}^{+}$ and $\mathbf U_{\mathcal H} = \bm 0$, which is equal to the MLE subject to the graph $\mathcal G^{+}_0 = (\bm X,\bm Y; \mathcal E^{+}\setminus \mathcal H, \mathcal I^{+})$.
Meanwhile, to test whether any edge in $\mathcal H$ exists, $\widehat{\bm\theta}^{(1)}$ is the MLE subject to an augmented graph $\mathcal G^{+}_1 = (\bm X, \bm Y; \mathcal E^{+}\cup\mathcal H,\mathcal I^{+})$ with hypothesized edges being added, namely, $\widehat{\mathbf U}^{(1)} = (\widehat{\mathbf U}^{(1)}_{\mathcal E^{+}\cup\mathcal H},\bm 0)$ and $\widehat{\mathbf W}^{(1)} = (\widehat{\mathbf W}^{(1)}_{\mathcal I^{+}},\bm 0)$.
Of note, since $\mathcal H$ is pre-specified by the user, $\mathcal G^{+}_1$ is not necessarily acyclic, and thus, not all edges in $\mathcal H$ could present in $\widehat{\mathbf U}^{(1)}$.
Furthermore, if a hypothesized edge $(k,j)$ is present in $\widehat{\mathbf U}^{(1)}$, then $\{(k,j)\}\cup\mathcal E^{+}$ must have no directed cycle and \eqref{equation:likelihood-ratio} is strictly positive (nondegenerate).
However, even if \eqref{equation:likelihood-ratio} does not degenerate to zero, its limiting distribution can be complicated when there exist multiple ways of augmenting $\mathcal G^{+}$ with the edges in $\mathcal H$ while maintaining the resulting graph as a DAG.
Therefore, a regularity condition for $\mathcal H$ is necessary to rule out intractable situations. 

On the ground of the foregoing discussion, we introduce the concepts of nondegeneracy and regularity to characterize the behavior of \eqref{equation:likelihood-ratio} as in \citet{li2023inference}.

\begin{definition}[Nondegeneracy and regularity with respect to $\mathcal G^+$] \label{def:testability} \ 
\begin{enumerate}[(A)]
  \item An edge $(k,j) \in \mathcal{H}$ is said to be nondegenerate with respect to an ancestral graph $\mathcal{G}^{+} = (\bm Y, \bm X; \mathcal{E}^{+}, \mathcal{I}^{+})$ if $\{(k,j)\} \cup \mathcal{E}^{+}$ contains no directed cycle. Otherwise, $(k,j)$ is said to be degenerate. Let $\mathcal{D} \subseteq \mathcal{H}$ be the set of all nondegenerate edges with respect to $\mathcal{G}^{+}$. A null hypothesis $H_0$ is said to be nondegenerate with respect to $\mathcal{G}^{+}$ if $\mathcal{D} \neq \emptyset$. Otherwise, $H_0$ is said to be degenerate.
  \item A null hypothesis $H_0$ is said to be regular with respect to $\mathcal G^{+}$ if $\mathcal{D} \cup \mathcal{E}^{+}$ contains no directed cycle. Otherwise, $H_0$ is called irregular.
\end{enumerate}
\end{definition}

Suppose $H_0$ is nondegenerate and regular. Then $\widehat{\bm\theta}^{(0)}$ is the MLE subject to the graph $\mathcal G^{+}_0 = (\bm X,\bm Y; \mathcal E^{+}\setminus \mathcal D, \mathcal I^{+})$ and $\widehat{\bm\theta}^{(1)}$ is the MLE subject to the graph  $\mathcal G^{+}_1 = (\bm X, \bm Y; \mathcal E^{+}\cup\mathcal D,\mathcal I^{+})$.

Now, we investigate the limiting distribution of \eqref{equation:likelihood-ratio} and derive an asymptotic test based on it. 
To this end, define the statistic
\begin{equation}\label{equation:test-statistic}
    T(\mathcal D) = \begin{cases}
        2 \left( L (\widehat{\bm{\theta}}^{(1)}, \widehat{\bm{\Omega}})- L (\widehat{\boldsymbol{\theta}}^{(0)}, \widehat{\bm{\Omega}}) \right) & \text{if } |\mathcal D| \text{ is fixed}, \\
        \left (2 \left( L (\widehat{\bm{\theta}}^{(1)}, \widehat{\bm{\Omega}})- L (\widehat{\boldsymbol{\theta}}^{(0)}, \widehat{\bm{\Omega}}) \right) - |\mathcal D|\right)/ \sqrt{2|\mathcal{D}|}  & \text{if } |\mathcal D| \to\infty.
    \end{cases}
\end{equation}

\begin{theorem}[Limiting distribution] 
\label{theorem:asymptotic}
Assume the null hypothesis $H_0$ is nondegenerate and regular. 
Suppose $\mathbb P(\widehat{\mathcal G}^{+} = {\mathcal G}^{+}) \to 1$ as $n\to\infty$. 
Then we have $\mathbb P(\widehat{\mathcal D}=\mathcal D)\to 1$.
{In addition, if $\|\widehat{\bm\Omega} - \bm\Omega\|_2^2 = O_{\mathbb P}({|S|\log(p\vee n)/n})$ where $S=\{(k,j) : \Omega_{kj}\neq 0 \}$,}
then under $H_0$,   
\begin{equation*}
    \begin{split}
        T(\widehat{\mathcal D}) \dto 
        \begin{cases}
        \chi^2_{|\mathcal D|}, & \text{if } |\mathcal D| \text{ is fixed and } {|S| \log (p \vee n)}/{n} \rightarrow 0, \\
        N(0,1), & \text{if } |\mathcal D| \to\infty \text{ and } { |{\mathcal{D}}|} |S| \log (p \vee n)/{n} \rightarrow 0.
        \end{cases}
    \end{split}
\end{equation*}
\end{theorem}

On the basis of Theorem \ref{theorem:asymptotic}, we conduct inference by substituting $|\mathcal D|$ by its estimate $|\widehat{\mathcal D}|$ and proceed with the  empirical rule: (1) use the chi-squared test when $|\widehat{\mathcal D}| < 50$, and (2) use the normal test when $|\widehat{\mathcal D}|\geq 50$. 

Theorem \ref{theorem:asymptotic} requires a good estimator $\widehat{\bm\Omega}$ of $\bm\Omega = \bm\Sigma^{-1}$ to account for the confounding effects, where $\bm\Sigma=\Cov(\bm\varepsilon)$. 
To estimate $\bm\Omega$, let $\widehat{\bm\varepsilon}_{+,i} = (\mathbf I - \widehat{\mathbf U})^\top \mathbf{Y}_{+,i} - \widehat{\mathbf W}^\top \mathbf X_{+,i}$; $i=1,\ldots,n$ be the estimated residuals after fitting \eqref{equation:model} with Algorithm \ref{algorithm:pruning}. 
Here we use the neighborhood selection method \citep{meinshausen2006high} with an additional refitting to obtain a positive definite estimate $\widehat{\bm\Omega}$.
In Supplementary Materials, we include the computational details and show that this estimator satisfies $\|\widehat{\bm\Omega}-\bm\Omega\|_F^2 = O_{\mathbb P}( |S| \log(p\vee n)/n)$ so that Theorem \ref{theorem:asymptotic} applies.

\begin{remark}\label{remark:irregular-hypothesis}
    In Theorem \ref{theorem:asymptotic}, we focus on nondegenerate and regular hypotheses. 
    For a degenerate case, we define the p-value as one. For an irregular case where $\mathcal{D}\cup \mathcal{E}^{+}$ contains a directed cycle, we decompose $H_0$ into sub-hypotheses $H_0^{(1)},\ldots,H_0^{(r)}$, each of which is regular. Then testing $H_0$ is reduced to multiple testing for $H_0^{(1)},\ldots,H_0^{(r)}$.
\end{remark}

Finally, we discuss two aspects of likelihood estimation and inference in the presence of unmeasured confounding. 
First, when $\bm\Sigma$ is non-diagonal, the likelihood in \eqref{equation:log-likelihood} cannot be factorized according to $\mathcal G$ (or $\mathcal G^+$). This implies that, unlike the case without latent confounders \citep{shojaie2010penalized}, the parameters of each equation in \eqref{equation:model} cannot be estimated separately given $\mathcal G^+$. Indeed, the likelihood estimation of $(\mathbf U,\mathbf W)$ in \eqref{equation:model} requires a preliminary estimate of $\bm\Omega$ to account for correlations arising from hidden confounding.
Furthermore, compared to \citet{li2023inference}, the likelihood ratio \eqref{equation:likelihood-ratio} is no longer a sum of likelihood ratios of equations associated with nondegenerate hypothesized edges, rendering inference more challenging in both computation and theory when hidden confounders are present. Computationally, the likelihood ratio \eqref{equation:likelihood-ratio} requires maximization of the full likelihood, which is costly for a large-scale graph. Theoretically, estimating $\bm\Omega$ and $(\mathbf U,\mathbf W)$ in high-dimensional situations may suffer from the curse of dimensionality. 

Second, to mitigate the challenges in inference, we may conduct inference with respect to a sub-DAG to achieve dimensionality reduction. Specifically, let $\mathcal D$ be the nondegenerate edges of $H_0$. 
Given ARG $\mathcal G^+$, we perform likelihood inference using a sub-DAG (of ARG) ${\mathcal G}^+_{\text{sub}} = ({\bm X}_{\text{sub}}, {\bm Y}_{\text{sub}}; {\mathcal E}_{\text{sub}}^+, {\mathcal I}_{\text{sub}}^+)$, where all edges specified in $\mathcal D$ are among primary variables ${\bm Y}_{\text{sub}}$, and ${\bm Y}_{\text{sub}}$ are non-descendants of $\bm Y\setminus {\bm Y}_{\text{sub}}$ in the graph $(\bm X,\bm Y; \mathcal E^+\cup\mathcal D, \mathcal I^+)$, ${\bm X}_{\text{sub}}$ is the set of intervention variables of ${\bm Y}_{\text{sub}}$, ${\mathcal E}_{\text{sub}}^+$ is the set of ancestral relations among ${\bm Y}_{\text{sub}}$, and ${\mathcal I}_{\text{sub}}^+$ is the set of interventional relations between ${\bm X}_{\text{sub}}$ and ${\bm Y}_{\text{sub}}$ in ARG $\mathcal G^+$. Then the test statistic \eqref{equation:test-statistic} is computed within the sub-DAG ${\mathcal G}_{\text{sub}}^+$, which reduces computation. Furthermore, Theorem \ref{theorem:asymptotic} holds true when the estimator of the smaller precision matrix ${\bm\Omega}_{\text{sub}}$ enjoys the desired convergence rate $O_{\mathbb P}(\sqrt{|{S}_{\text{sub}}|\log({p}_{\text{sub}}\vee n)/n})$ in operator norm, where the subscript $_{\text{sub}}$ denotes the quantities corresponding to the structural equations of ${\bm Y}_{\text{sub}}$.

\section{Theory}\label{section:causal-discovery-theory}

In this section, we develop a theory to quantify the finite sample performance as well as the complexities of Algorithms \ref{algorithm:peeling}--\ref{algorithm:pruning} when TLP is used for computation. 

To proceed, we introduce some technical conditions for casual discovery consistency. 
For $(k,j)\in\mathcal E^{+}$, let $\widetilde{\bm \Sigma}^{(k,j)}$ be the covariance matrix of $(\E(Y_k \mid \bm Y_{\nui_{\mathcal G}(k,j)}, \bm X), \bm Y_{\nui_{\mathcal G}(k,j)}, \bm X)$.
Moreover, let $s = \max_{(k,j)\in\mathcal E^{+}} (\kappa+\nu_2, \nu_1) \vee \max_{1\leq k\leq p}\|\mathbf{V}_{+,k}\|_0$ be the maximum sparsity-level in the estimation procedure, {where $\nu_1,\nu_2,\kappa$ depends on $(k,j)$ which is dropped for conciseness.} 
Assume there exist constants $c_0,c_1,c_2,c_3>0$ such that 
    \begin{enumerate}
        \item [(C1)]
        $\min_{(k,j)\in\mathcal E^{+}}\min_{B:|B|\leq 2s}\min_{\bs v: \|\bs v\|_2= 1, \|\bs v_{B^c}\|_1\leq 3\|\bs v_{B}\|_1 + c_0s\sqrt{\log(p)/n}} \langle \bs v, \widetilde{\bm\Sigma}^{(k,j)}\bs v \rangle \geq c_1$.
        
        \item [(C2)] $\min_{\mathrm V_{kj}\neq 0} |\mathrm V_{kj}| \geq c_2 \sqrt{{\log(q\vee n)}/{n}}$.

        \item [(C3)] $\min_{\mathrm U_{kj}\neq 0} |\mathrm U_{kj}| \geq c_3 \sqrt{\log(p\vee n)/n}$.

        \item [(C4)] $\max_{1\leq k\leq p} \{ |\an_{\mathcal G}(k)|,|\inter_{\mathcal G}(k)|,\|\mathbf U_{+,k}\|_1\} = O(1)$, and $\max_{(k,j)\in\mathcal E^{+}}(\Diag(\widetilde{\bm\Sigma}^{(k,j)})) = O(1)$.
    \end{enumerate}
Condition (C1) is a restricted eigenvalue condition, which is common in high-dimensional estimation \citep{bickel2009simultaneous} and can be viewed as a stronger version of (A1) in Theorem \ref{theorem:identifiability}.  (C2) and (C3) impose restrictions on the minimal signal strengths of $\mathbf V$ and $\mathbf U$ so that the ARG $\mathcal G^{+}$ and DAG $\mathcal G$ can be consistently recovered, respectively. They are similar to the beta-min condition \citep{meinshausen2006high} and the degree of separation condition \citep{shen2012likelihood} in the variable selection literature.

\begin{theorem}\label{theorem:consistency}
    Suppose Assumptions (A1)--(A3) in Theorem \ref{theorem:identifiability} are satisfied and assume $\bm X$ is sub-Gaussian with mean zero and parameter $\varsigma^2$.
    \begin{enumerate}[(A)]
        \item {(Parameter estimation)} Suppose (C1), (C2), (C4) are met with sufficiently large $c_0,c_1,c_2$. 
        Suppose the tuning parameters are suitably chosen such that 
        \begin{enumerate}[(i)]
            \item In Algorithm \ref{algorithm:peeling}, $0.01c_2\sqrt{\log(q\vee n)/n} \leq \tau' \leq 0.4\min_{\mathrm V_{kj}\neq 0}|\mathrm V_{kj}|$, $\kappa'_j = \|\bs V_{+,j}\|_0$ for $1\leq j\leq p$.

            \item In Algorithm \ref{algorithm:pruning}, $0.5 c_3 \sqrt{{\log(p\vee n)}/{n}} \leq \tau$, $\nu_1 = \lceil \operatorname{TLP}_{\tau}((\bs \alpha_1,\bs \alpha_2))\rceil$, $\nu_2 = \lceil \operatorname{TLP}_{\tau}(\bm\gamma)\rceil$, and $\kappa = \lceil \operatorname{TLP}_{\tau}(\bm\beta)\rceil$ for any $(k,j)\in\mathcal E^{+}$.
        \end{enumerate}
        Then there exists constant $C_1 > 0$ such that when $n$ is sufficiently large
        \begin{equation*}
            |\widehat{\mathrm{U}}_{kj} - {\mathrm{U}_{kj}}| \leq C_1 \sqrt{{\log(p\vee n)}/{n}},
        \end{equation*}
         almost surely under $\mathbb P_{(\mathbf U,\mathbf W,\bm\Sigma)}$. Moreover, Algorithms \ref{algorithm:peeling} and \ref{algorithm:pruning} respectively terminate in $O(p\times \log(s)\times (q^3 + nq^2) )$ and $O(|\mathcal E^{+}|\times \log(s) \times (q^3 + nq^2))$ operations almost surely.
        
        \item {(Graph recovery)} Additionally, if (C3) is satisfied with $c_3 > C_1 > \tau$, then when $n$ is sufficiently large we have $\widehat{\mathcal G} = \mathcal G$ almost surely.
        
    \end{enumerate}
\end{theorem}

By Theorem \ref{theorem:consistency}, the proposed method achieves causal discovery consistency in terms of consistent parameter estimation and structure recovery. 
Moreover, Algorithms \ref{algorithm:peeling}--\ref{algorithm:pruning} enjoy low-order polynomial time complexity almost surely provided that the data are randomly sampled from \eqref{equation:model}.

\section{Numerical examples}\label{section:numerical}

\subsection{Simulations}

This subsection investigates via simulations the operating characteristics of GrIVET, including the qualities of structure learning, parameter estimation, and statistical inference.

To generate an observation $(\bm Y,\bm X)$, we first introduce hidden variables $\bm \eta \sim N(\bm 0, \mathbf{I}_{r\times r})$ as unmeasured confounders. 
Then, we sample $\bm X$ from $N(\bm 0, \mathbf{I}_{q \times q})$ for continuous interventions or from $\{-1,1\}^q$ with equal probability for discrete interventions. 
Given $\bm X$ and $\bm\eta$, we generate $\bm Y$ according to 
\begin{equation}\label{equation:simulation}
\bm Y=\mathbf{U}^{\top} \bm Y+\mathbf{W}^{\top} \bm X+ \bm \Phi^{\top} \bm \eta + \bm e, \quad \bm e \sim N\left(\mathbf{0}, \Diag\left(\sigma_{1}^{2}, \ldots, \sigma_{p}^{2}\right)\right).
\end{equation}
We conduct simulations with the following settings. 
\begin{itemize}
  \item \textbf{Hub graph.} Let $p=101$, $q=252$, and $r=10$. 
  For $\mathbf U$, $(\mathrm{U}_{1,j})_{2\leq j\leq p}$ are independently sampled from $\{-1,1\}$ with equal probability,  while the rest are set to $0$. This generates a sparse graph with the dense neighborhood of the first node. 
  Let $\mathbf{W}_{q\times p}=(\mathbf{I}_{p \times p},\mathbf{I}_{p \times p}, \mathbf{F}^{\top})^{\top}$ where the entries $(\mathrm{F}_{j,2j}, \mathrm{F}_{j,2j+1})_{1\leq j\leq q-2p}$ are set to $1$, while other entries of $\mathbf{F}$ are zero. Then $X_{j},X_{2j}$ are IVs of $Y_j$ for $j=1,\ldots,p$ and $X_{2p+1},\ldots,X_{q}$ are invalid IVs with two intervention targets. 
  For the confounders, $\mathrm{\Phi}_{1,1}$ and $(\mathrm{\Phi}_{jk})^{1\leq j\leq r}_{10j-8\leq k\leq 10j+1}$ are sampled uniformly from $(-0.4,-0.6)\cup(0.4,0.6)$, while other entries of $\mathbf{\Phi}$ are zero. We generate $(\sigma_{1}, \ldots, \sigma_{p})$ uniformly from $(0.4,0.6)$.

  \item \textbf{Random graph.} Let $p=100$, $q=250$, and $r=10$. 
  For $\mathbf U$, the upper off-diagonals $(\mathrm{U}_{kj})_{k<j}$ are sampled independently from $\{0,1\}$ according to $\text{Bernoulli}({1}/{10p})$ while other entries are zero. Set $\mathbf{W}_{q\times p}=(\mathbf{I}_{p \times p},\mathbf{I}_{p \times p}, \mathbf{F}^{\top})^{\top}$ where $(\mathrm{F}_{j,2j-1},\mathrm{F}_{j,2j})_{1\leq j \leq 1-2p}$ are set to $1$, while other entries of $\mathbf{F}$ are zero. Then $X_{j},X_{2j}$ are IVs of $Y_j$ for $j=1,\ldots,p$ and $X_{2p+1},\ldots,X_{q}$ are invalid IVs with two intervention targets. For the confounders, $(\mathrm{\Phi}_{jk})^{1\leq j\leq r}_{10j-9\leq k\leq 10j}$ are sampled uniformly from $(-0.4,-0.6)\cup(0.4,0.6)$, while other entries of $\mathbf{\Phi}$ are zero. We generate $(\sigma_{1}, \ldots, \sigma_{p})$ uniformly from $(0.4,0.6)$.
\end{itemize}

\paragraph*{Structure learning.}
After obtaining ancestral relations from Algorithm \ref{algorithm:peeling}, we implement Algorithm \ref{algorithm:pruning} to confirm parental relations but with constraints also imposed on the parameter of interest. Four graph metrics are used for evaluation: the false discovery rate (FDR), the true positive rate (TPR), the Jaccard index (JI), and the structural Hamming distance (SHD). 
The results in Table \ref{table:part_1_2} demonstrate the strong performance of GrIVET in structure learning. Note that a high TPR indicates GrIVET's capability to detect the true existing edges, while the FDR remains low, signifying the high specificity of GrIVET. In Supplementary Materials Section 3.3, we further compare GrIVET with RFCI \citep{colombo2012learning} and LRpS-GES \citep{frot2019robust} in terms of structural learning accuracy. GrIVET compares favorably against the competitors.

\begin{table}[ht] 
\footnotesize
\centering
\caption{False discovery rate (FDR), true positive rate (TPR), structural Hamming distance (SHD), and Jaccard index (JI) of GrIVET for causal discovery over 1000 simulation replications. To compute the metrics, let TP, RE, FP, and FN be the numbers of identified edges with correct directions, those with wrong directions, estimated edges not in the skeleton of the true graph, and missing edges compared to the true skeleton. Then $\mathrm{FDR}=(\mathrm{RE}+\mathrm{FP}) / (\mathrm{TP}+\mathrm{RE}+\mathrm{FP})$, $\mathrm{TPR}=\mathrm{TP} /(\mathrm{TP}+\mathrm{FN})$, $\mathrm{SHD}=\mathrm{FP}+\mathrm{FN}+\mathrm{RE}$, and $\mathrm{JI}=\mathrm{TP} /(\mathrm{TP}+\mathrm{SHD})$.}
\label{table:part_1_2}
\begin{tabular}{l l c c c c c}
\hline \multicolumn{1}{l}{Graph} & \multicolumn{1}{l}{Intervention} & \multicolumn{1}{c}{$n$} & \multicolumn{1}{c}{FDR($\%$)} & \multicolumn{1}{c}{TPR($\%$)} & \multicolumn{1}{c}{SHD} & \multicolumn{1}{c}{JI($\%$)} \\ 
\hline
 Hub  & Continuous & 500 &   0.000 & 100.000 & 0.000 & 100.000 \\
     &            &  400 &   0.000 &  99.998 & 0.002 &  99.998 \\
     &            &  300 &   0.000 &  99.998&  0.002 &  99.998 \\
     & Discrete   & 500 &   0.000 &  99.999 & 0.001 &  99.999 \\
     &            &  400 &   0.000 &  99.998 & 0.002 &  99.998 \\
     &            &  300 &   0.000 &  99.999 & 0.001 &  99.999 \\ \hline
Random & Continuous   & 500 &   0.011 &   98.600 & 0.001 &  98.589 \\
      &            &  400 &   0.000 &   98.600 &  0.000 &  98.600 \\
     &            &  300 &   0.018 &   98.590 &  0.003 &  98.575 \\
     & Discrete   & 500 &   0.000 &   98.600 &  0.000 &  98.600 \\
     &            &  400 &   0.024 &   98.600 &  0.002 &  98.576 \\
     &            &  300 &   0.000 &   98.600 &  0.000 &  98.600 \\
\hline
\end{tabular}
\end{table}

\paragraph*{Parameter estimation.}

\begin{table}[ht] 
    \footnotesize
    \centering
    \caption{Parameter estimation: the average of largest absolute difference (Max AD), the average absolute differences (Mean AD), and the average squared differences (Mean SqD) between the estimated parameters and the true parameters for two competing methods over 1000 simulation replications.}
    \label{table:part_2_2}
    \begin{tabular}{l l c c c c}
      \hline \multicolumn{1}{l}{Graph} & \multicolumn{1}{l}{Intervention} & \multicolumn{1}{c}{$n$} & \multicolumn{1}{c}{GrIVET} & & \multicolumn{1}{c}{Direct regression \citep{li2023inference}}\\ 
  & & & (Max AD, Mean AD, Mean SqD)& & (Max AD, Mean AD, Mean SqD) \\
    \hline
    Hub  & Continuous & 500 & (0.06107, 0.01808, 0.00052) & & (0.12817, 0.02448, 0.00142)\\
       &            &  400 & (0.06863, 0.02037, 0.00066) & & (0.13196, 0.02637, 0.00156)\\
       &            &  300 & (0.07922, 0.02347, 0.00087) & & (0.13395, 0.02873, 0.00170)\\
       & Discrete   & 500 & (0.06119, 0.01803, 0.00051) & & (0.12770, 0.02434, 0.00141)\\
       &            &  400 & (0.06932, 0.02030, 0.00065) & & (0.13041, 0.02621, 0.00153)\\
       &            &  300 & (0.08046, 0.02355, 0.00088) & & (0.13334, 0.02867, 0.00169)\\ \hline
  Random & Continuous   & 500 & (0.02836, 0.01445, 0.00034) & & (0.04254, 0.01791, 0.00076)\\
       &            &  400 & (0.03245, 0.01660, 0.00045) & & (0.04390, 0.01899, 0.00079)\\
       &            &  300 & (0.03760, 0.01939, 0.00060) & & (0.04709, 0.02150, 0.00091)\\
       & Discrete   & 500 & (0.02910, 0.01505, 0.00037) & & (0.04287, 0.01808, 0.00075)\\
       &            &  400 & (0.03272, 0.01686, 0.00046) & & (0.04432, 0.01962, 0.00081)\\
       &            &  300 & (0.03619, 0.01879, 0.00057) & & (0.04756, 0.02146, 0.00094)\\
    \hline
    \end{tabular}
  \end{table}

We compare the proposed IV estimation method in Section \ref{section:known-order-estimation} with the regression method without any adjustment for confounding \citep{li2023inference}. To evaluate the quality of estimation, we consider three metrics, the average maximum absolute deviation, the mean absolute deviation, and the mean square deviation between true coefficients and estimates over 1000 runs.
As demonstrated in Table \ref{table:part_2_2}, GrIVET enhances parameter estimation by accounting for latent confounding. As anticipated, GrIVET's estimation improves with increasing sample size $n$, while the naive regression method \citep{li2023inference} remains inconsistent. Furthermore, GrIVET's advantages become more pronounced when stronger confounding effects are present, as evidenced by additional simulations in the Supplementary Materials.

\newpage
\paragraph*{Inference.}

We now evaluate the empirical performance of the proposed tests in terms of size and power.
For the empirical size, we calculate the percentage of times $H_0$ is rejected out of 1000 simulations when $H_0$ is true. For the power, we consider three alternative hypotheses $H_a$, where all the edges in $H_0$ exist. The empirical power of a test is the percentage of times $H_0$ is rejected out of 1000 simulations when $H_a$ is true.
The adjacency matrix $\mathbf U$ is modified according to the null and alternative hypotheses.

\begin{itemize}
  \item \textbf{Hub graph, fixed $\mathcal H$.} For the size, consider $\mathcal{H} = \{(2,7)\}$, $\mathcal{H} = \{(2,7),(7,12),(12,17)\}$, and $\mathcal{H} = \{(2,7),(7,12),(12,17),(17,22),(22,27)\}$. For the power, consider $\mathcal{H} = \{(1,2)\}$, $\mathcal{H} = \{(1,2),(1,12),(1,22)\}$, and $\mathcal{H} = \{(1,2),(1,12),(1,22),(1,32),(1,42)\}$.
  \item \textbf{Random graph, fixed $\mathcal H$.} We consider $\mathcal{H} = \{(1,6)\}$, $\mathcal{H} = \{(1,6),(6,11),(11,16)\}$, and $\mathcal{H} = \{(1,6),(6,11),(11,16),(16,21),(21,26)\}$ for both size and power.

  \item \textbf{Random graph, random $\mathcal H$.} 
  We also consider testing 50 randomly selected edges individually. Here, a random graph is generated so that 20 of these selected edges are present in the true DAG (i.e., $H_a$ is valid). As a result, for every selected edge, $H_0$ holds in roughly $600$ repetitions and $H_a$ holds in roughly $400$ repetitions.
  
\end{itemize}

As shown in Table \ref{table:part_3_3} for fixed $\mathcal H$, empirical sizes are close to the nominal $\alpha=0.05$ under $H_0$, and the proposed test enjoys desirable power under $H_a$. 
Figure \ref{Fig:random_edges} presents similar results for testing random $\mathcal H$. The Supplementary Materials display that the sampling distribution of the test statistic is close to the derived asymptotic distribution in Theorem \ref{theorem:asymptotic}. Additional simulation details and results are also available in Supplementary Materials.

\begin{table}[ht] 
    \footnotesize
    \centering
    \caption{Empirical size for GrIVET at nominal level $\alpha=0.05$, respectively for $|\mathcal{D}|=1$, $|\mathcal{D}|=3$ and $|\mathcal{D}|=5$, over 1000 simulation replications.}
    \label{table:part_3_3}
    \begin{tabular}{l l c c c}
    \hline \multicolumn{1}{l}{Graph} & \multicolumn{1}{l}{Intervention} & \multicolumn{1}{c}{$n$} & \multicolumn{1}{c}{Size ($|\mathcal{D}|=1,3,5$)} & \multicolumn{1}{c}{Power ($|\mathcal{D}|=1,3,5$)}\\ \hline 
    Hub  & Continuous &  500 & (0.028,0.026,0.029) & (1.000,1.000,1.000) \\
         &            &  400 & (0.043,0.038,0.035) & (1.000,1.000,1.000) \\
         &            &  300 & (0.037,0.030,0.034) & (1.000,1.000,1.000) \\
         & Discrete   &  500 & (0.036,0.040,0.027) & (1.000,1.000,1.000) \\
         &            &  400 & (0.051,0.040,0.040) & (1.000,1.000,1.000) \\
         &            &  300 & (0.052,0.041,0.035) & (1.000,1.000,1.000) \\ \hline
    Random & Continuous   & 500 & (0.038,0.037,0.026) & (1.000,1.000,1.000) \\
         &            &  400 & (0.033,0.031,0.028) &   (1.000,1.000,1.000) \\
         &            &  300 & (0.033,0.025,0.030) &   (1.000,1.000,1.000) \\
         & Discrete   &  500 & (0.040,0.029,0.027) &   (1.000,1.000,1.000) \\
         &            &  400 & (0.042,0.034,0.040) &   (1.000,1.000,1.000) \\
         &            &  300 & (0.029,0.033,0.034) &   (1.000,1.000,1.000)\\
    \hline
    \end{tabular}
    \end{table}
    
    \begin{figure}[H]
        \centering
        \begin{subfigure}{.4\textwidth}
            \includegraphics[width=1\textwidth]{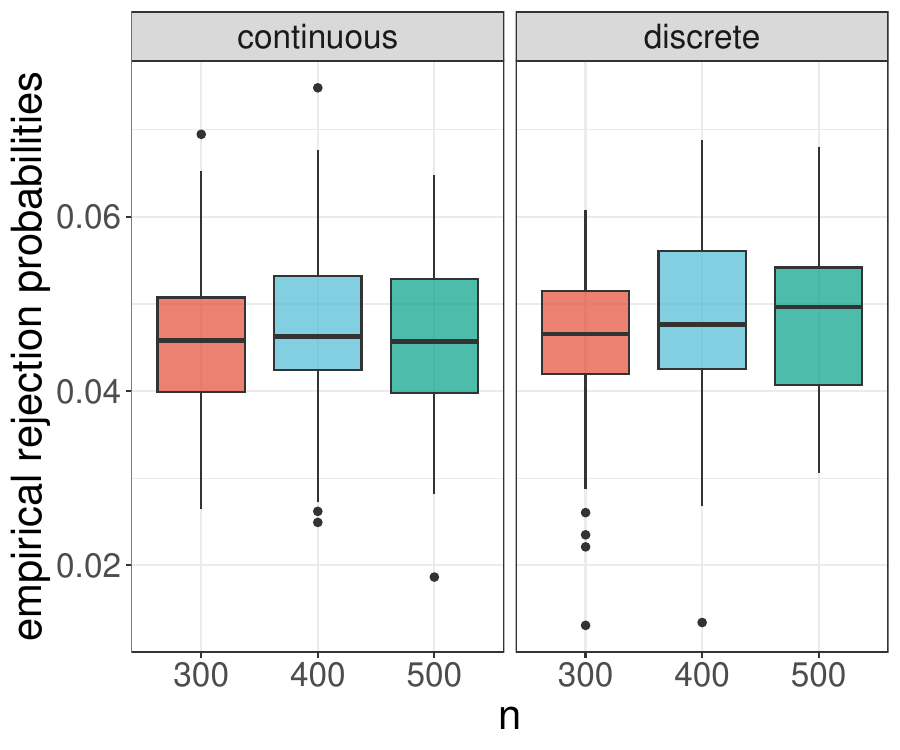}
            \caption{$H_0$ holds}
        \end{subfigure}
        \begin{subfigure}{.4\textwidth}
            \includegraphics[width=1\textwidth]{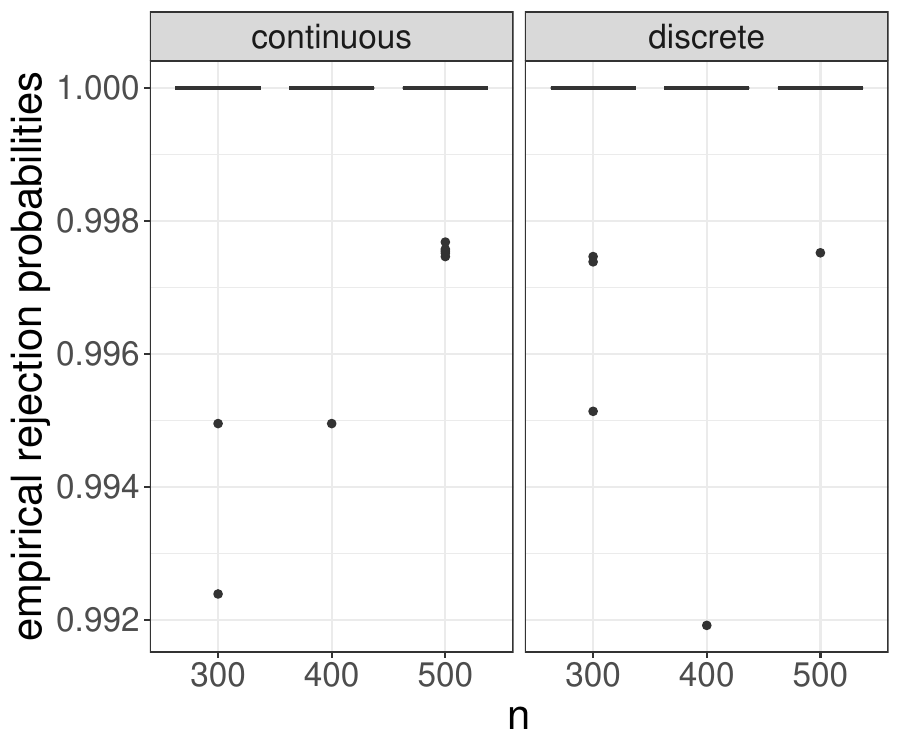}
            \caption{$H_a$ holds}
        \end{subfigure}
        \caption{The boxplots of the empirical rejection probabilities for testing randomly selected edges. The nominal level is $\alpha = 0.05$.} 
        \label{Fig:random_edges}
    \end{figure}

\subsection{ADNI data analysis}

In this subsection, GrIVET is applied to analyze the Alzheimer's Disease Neuroimaging Initiative (ADNI) dataset (available at \url{https://adni.loni.usc.edu}). The goal is to infer gene pathways related to Alzheimer's Disease (AD) in order to elucidate the gene-gene interactions in  AD/cognitive impairment patients and healthy individuals, respectively.

\paragraph*{Dataset.}
The dataset comprises gene expression levels adjusted for five covariates: gender, handedness, education level, age, and intracranial volume.
For data analysis, we select genes with at least one SNP at a marginal significance level below $10^{-14}$, resulting in $p=21$ genes as primary variables.
For these genes, we further extract their marginally most correlated two SNPs, yielding $q=42$ SNPs as unspecified intervention variables for subsequent data analysis. All gene expression levels are normalized.

The dataset initially categorizes individuals into four groups: Alzheimer's Disease (AD), Early Mild Cognitive Impairment (EMCI), Late Mild Cognitive Impairment (LMCI), and Cognitive Normal (CN). For our analysis, we treat 247 CN individuals as controls and the remaining 462 individuals as cases (AD-MCI). We then use the gene expressions and the SNPs to infer gene pathways for the 462 AD-MCI and 247 CN control cases, respectively.

\begin{figure}[H]
    \centering
    \begin{subfigure}{.4\textwidth}
        \includegraphics[width=1\textwidth]{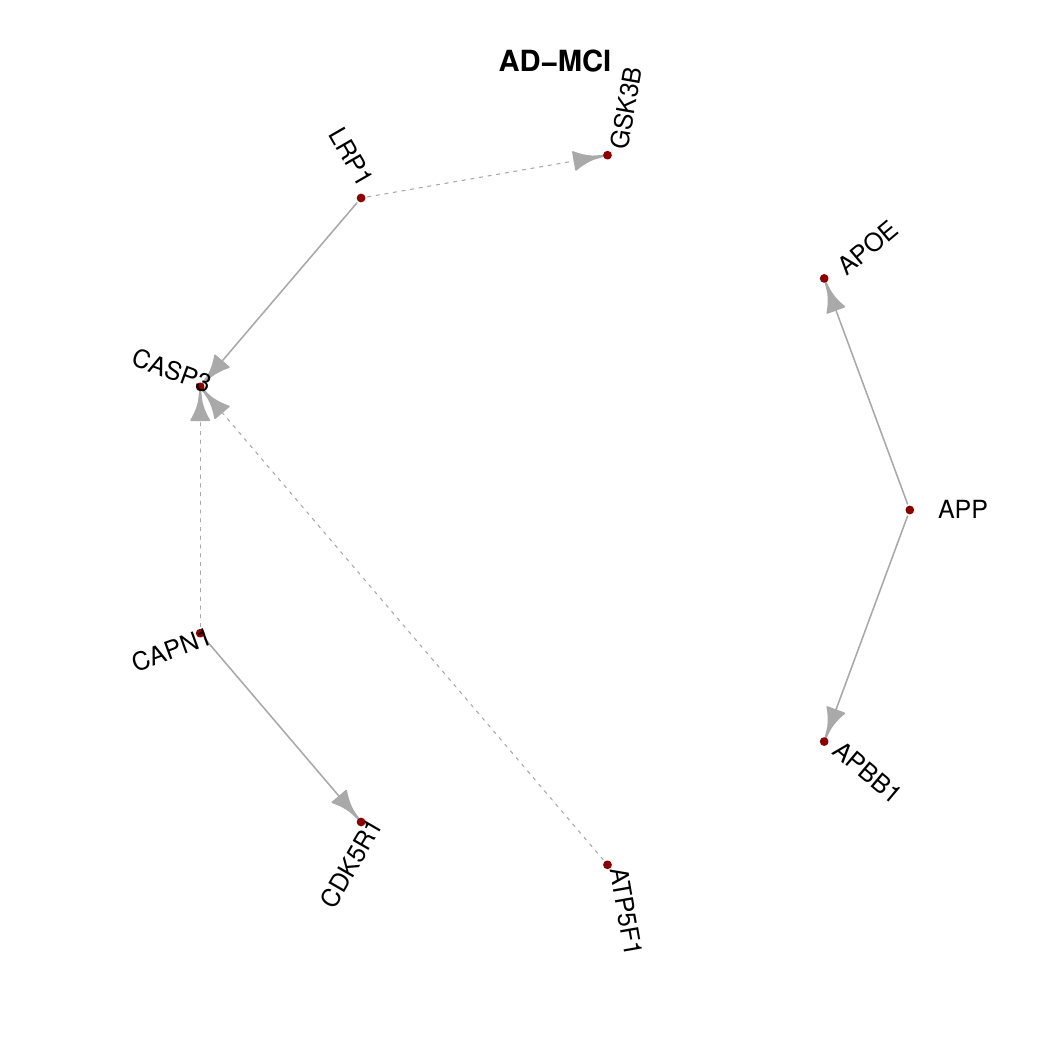}
        \caption{AD-MCI}
    \end{subfigure}
    \begin{subfigure}{.4\textwidth}
        \includegraphics[width=1\textwidth]{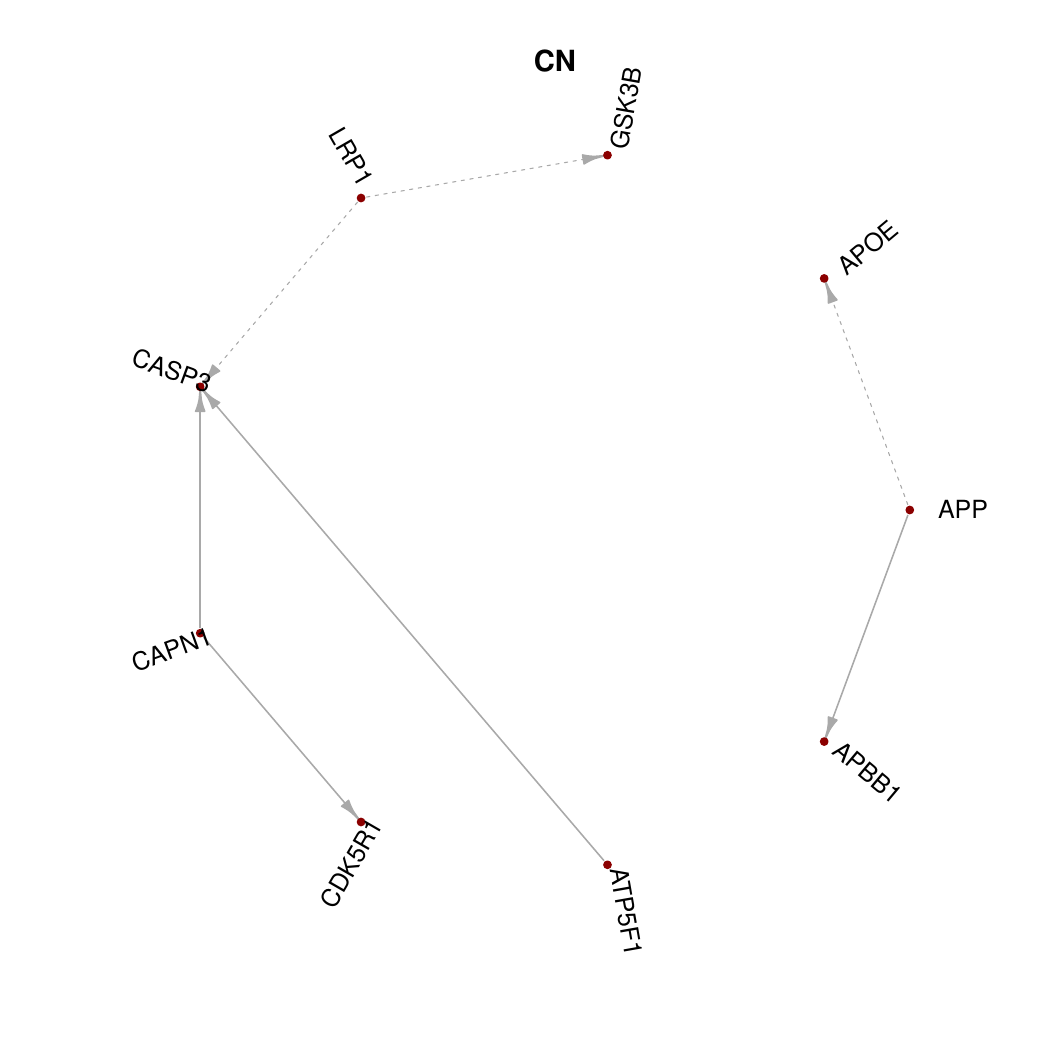}
        \caption{CN}
    \end{subfigure}
    \caption{Display of the genes associated with proposed tests. (a) and (b): Solid/dashed arrows indicate significant/insignificant edges at $\alpha=0.05$ after adjustment for multiplicity by the Bonferroni-Holm correction.}
    \label{Fig:ADNI1}
\end{figure}

\paragraph*{Hypotheses.}
We focus on statistical inferences related to genes APP and CASP3 \citep{julia2017genetics,su2001activated}. As in Figure \ref{Fig:ADNI1}, for each edge $(k,j)$, we consider testing
$H_{0}: \mathrm{U}_{k j}=0 \text { versus } H_{a}: \mathrm{U}_{k j} \neq 0$.

\paragraph*{Results.}
Figure \ref{Fig:ADNI1} displays the p-values and significant results under the level $\alpha=0.05$ after the Holm-Bonferroni adjustment for $2\times7=14$ tests. The tests exhibit strong evidence for the presence of $\{\mathrm{LRP1} \rightarrow \mathrm{CASP3}, \ \mathrm{APP} \rightarrow \mathrm{APOE}\}$ in the AD-MCI group, but no evidence in the CN group. Meanwhile, this result suggests the presence of connections $\{\mathrm{CAPN1} \rightarrow \mathrm{CASP3}, \ \mathrm{ATP5F1} \rightarrow \text {CASP3}\}$ in the CN group but not so in the AD-MCI group. In both groups, we identify directed connection $\mathrm{APP} \rightarrow \text {APBB1}$. Figure \ref{Fig:resi_corr} shows the residual correlation matrices for both groups, suggesting the existence of unmeasured confounding. The Supplementary Materials include normal Q-Q plots of residuals, demonstrating that the normality assumption is approximately satisfied for both groups.

\begin{figure}[H]
    \centering
    \begin{subfigure}{.4\textwidth}
        \includegraphics[width=1\textwidth]{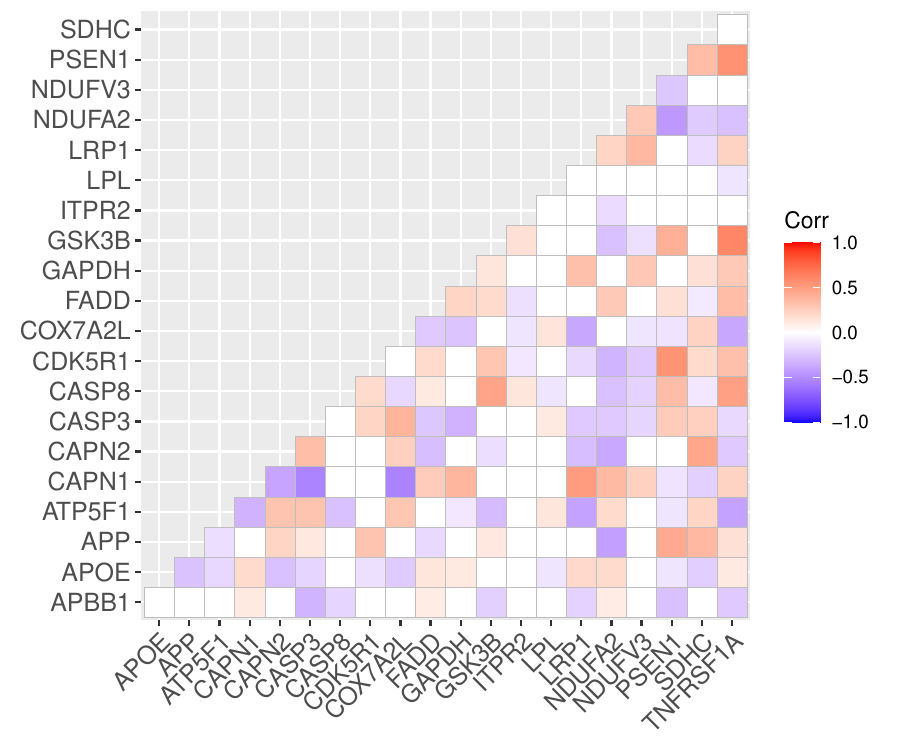}
        \caption{AD-MCI}
    \end{subfigure}
    \begin{subfigure}{.4\textwidth}
        \includegraphics[width=1\textwidth]{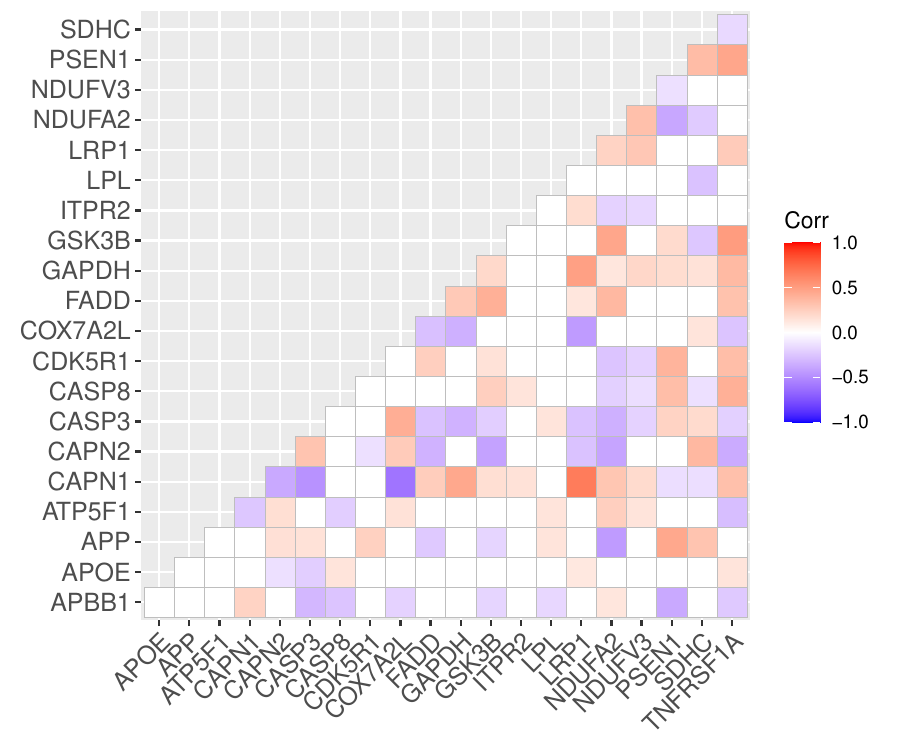}
        \caption{CN}
    \end{subfigure}
    \caption{Display of residual correlation matrices for AD-MCI and CN groups.} 
    \label{Fig:resi_corr}
\end{figure} 

Some of our discoveries agree with the existing findings. Specifically, our result indicates the presence of connection APP $\rightarrow$ APOE for the AD-MCI group, but not for the CN group, which seems consistent with the knowledge that APP and APOE are functionally linked in brain cholesterol metabolism \citep{liu2017effective} and the contributions of APOE to the pathophysiology of AD \citep{bu2009apolipoprotein}. The connection LRP1 $\rightarrow$ CASP3 also differs in AD-MCI and CN groups, which may serve to support the conclusion that activated CASP3 may be a factor in functional decline and may have an important role in neuronal cell death and plaque formation in AD brain \citep{su2001activated} given the finding that both APOE and its receptor LRP1 are present in amyloid plaques \citep{poirier1996apolipoprotein}. Moreover, the connection CAPN1 $\rightarrow$ CDK5R1 discovered in both groups can be found in the AlzNet database (interaction ID 24614).

\section{Discussion}\label{section:discussion}

This article proposes a novel instrumental variable procedure that integrates causal discovery and inference for a Gaussian directed acyclic graph with hidden confounders. One future research direction is to develop methodologies for analyzing  discrete/mixed-type (primary variable) data. Additionally, the present work uses individual-level data from a single study for causal discovery and inference. In many real applications, due to privacy concerns and ownership restrictions, the data are only available in the form of summary statistics (e.g., GWAS summary data) or in other privatized forms. Extending GrIVET to leverage these data is an important topic. Furthermore, multisource/decentralized data are ubiquitous, raising new challenges in communication, privacy, and handling of corrupted data. It would be promising to employ modern machine learning techniques, such as federated learning \citep{xiong2021federated,gao2021feddag}, to address these challenges and fully unleash the potential of large-scale causal discovery and inference.

Finally, we discuss two limitations of the present work. 
\begin{itemize}
    \item GrIVET necessitates the availability of valid IVs for each primary variable due to the hardness of causal identification in the presence of hidden confounding. In genetic research, there is an ample supply of genetic variants (e.g., SNPs) serving as IVs. Nonetheless, obtaining valid IVs can be challenging in certain applications. It is thus crucial to investigate the potential for causal discovery even when faced with an insufficient number of IVs.

    \item For inference, Theorem \ref{theorem:asymptotic} requires that $\mathbb P(\widehat{\mathcal G}^+=\mathcal G^+)\to 1$, which is guaranteed by Condition (C2) in Theorem \ref{theorem:consistency}. Fulfilling this requirement can be challenging; in such cases, one might turn to the post-selection inference framework \citep{berk2013valid} by concentrating on the parameters within the selected model. However, the test results should be meticulously interpreted, as these parameters cease to be causal or structural \citep{berk2013valid} unless $\mathbb P(\widehat{\mathcal G}^+=\mathcal G^+)\to 1$. In essence, (C2) enables the causal meaning of the tested parameters to be carried over to finite-sample inference. Exploring ways to lift the signal strength condition while preserving the causal interpretation for statistical inference after DAG structure learning \citep{wang2023confidence} is an important research topic.
\end{itemize}

\appendix

\section{Appendix}

\paragraph{Definition of d-separation \citep{pearl2009causality}.}

Consider a DAG $\mathcal G$ with node variables $(Z_1,\ldots,Z_{d})^\top$. 
Nodes $Z_k$ and $Z_j$ are adjacent if $Z_k\to Z_j$ or $Z_k\leftarrow Z_j$.
An undirected path between $Z_k$ and $Z_j$ in $\mathcal G$ is a sequence of distinct nodes $(Z_k,\ldots,Z_j)$ such that all pairs of successive nodes in the sequence are adjacent. 
A non-endpoint node $Z_m$ on an undirected path $(Z_k,\ldots,Z_{m-1},Z_m,Z_{m+1},\ldots,Z_j)$ is called a collider if $Z_{m-1}\to Z_m\leftarrow Z_{m+1}$. Otherwise, it is called a {non-collider}.
Let $A\subseteq \{1,\ldots,d\}$, where $A$ does not contain $k$ and $j$.
Then $\bm Z_A$ is said to block an undirected path $(Z_k,\ldots,Z_j)$ if at least one of the following holds: (1) the undirected path contains a non-collider that is in $\bm Z_A$, or (2) the undirected path contains a collider that is not in $\bm Z_A$ and has no descendant in $\bm Z_A$.
A node $Z_k$ is d-separated from $Z_j$ given $\bm Z_A$ if $\bm Z_A$ block every undirected path between $Z_k$ and $Z_j$; $k\neq j$.

\paragraph{Additional discussion of Figure \ref{fig:method} (a).}
Let $(k,j)\in\mathcal E^+$ and suppose all IVs are valid. 
We explain why $\bm X_{\ca(k)}$ may not be valid IVs after conditioning on $\bm Y_{\an(j)\setminus\{ k\} }$, as mentioned in Section 3.3. 
Let $l\in\ca(k)$ and $m\in\me(k,j)$ such that $Y_k$ is an unmediated parent of $Y_m$.
Note that in Figure 1 (a) of the main text, whenever $\eta\to Y_m$, then $\bm Y_{\an(j)\setminus\{ k\} }$ does not d-separate $\bm X_{\ca(k)}$ and $\eta$, since $Y_{m}$ is a collider in the undirected path $(X_{l}, Y_k, Y_{m}, \eta, Y_j)$.
As a result, $\bm X_{\ca(k)}$ and $\eta$ can be associated conditioned on $\bm Y_{\an(j)\setminus\{ k\} }$.

\paragraph{Additional discussion on identification of $\mathbf U$.} 
We have the following result.

\begin{lemma}\label{lemma:linear}
In \eqref{equation:model}, assume $\bm X$ and $\bm\varepsilon$ are independent.  
\begin{enumerate}[(A)]
    \item $\mathbb E(Y_k \mid \bm Y_{\nui(k,j)}, \bm X)$ is a linear combination of $(\bm Y_{\nui(k,j)}, \bm X)$.
    \item $\mathbb E (\varepsilon_j \mid \bm Y_{\nui(k,j)}, \bm X )$ is a linear combination of $(\bm Y_{\nui(k,j)}, \bm X_{\ca(k)^c})$.
\end{enumerate}
\end{lemma}

\begin{proof}
Here, (A) follows directly from \eqref{equation:model}. For (B), we have 
\begin{equation*}
    \mathbb E (\varepsilon_j \mid \bm Y_{\nui(k,j)}, \bm X ) = \mathbb E(\varepsilon_j \mid \bm \varepsilon_{\nui(k,j)}, \bm X) = \mathbb E(\varepsilon_j \mid \bm\varepsilon_{\nui(k,j)}) = \bm\pi^\top\bm\varepsilon_{\nui(k,j)},
\end{equation*}
where the last equality is due to the normality of $\bm\varepsilon$.
Finally, in \eqref{equation:model}, we immediately have $\bm\varepsilon_{\nui(k,j)}$ is linear in $(\bm Y_{\nui(k,j)}, \bm X_{\ca(k)^c})$.
\end{proof}

Now, we show that $\Cov(\bm\varepsilon,\bm X)=\bm 0$ is sufficient to derive the identification results in Section \ref{section:known-order-estimation}.
Given random variables $\zeta$ and $\bm \xi$, 
let $\mathbb L( \zeta \mid \bm\xi )$ be the best linear approximation of $\zeta$ using $\bm\xi$, namely $\mathbb L(\zeta\mid \bm \xi) = \widetilde{\bm\omega}^\top \bm\xi$ where
\begin{equation*}
  \widetilde{\bm\omega} = \argmin_{\bm\omega} \ \E(\zeta - \bm\omega^\top\bm\xi)^2.
\end{equation*}
For random variables $\zeta$, $\zeta'$, and $\bm\xi$, we have that (a) $\mathbb{L}(\zeta + \zeta'\mid \bm\xi) = \mathbb{L}(\zeta \mid \bm\xi) + \mathbb{L}(\zeta'\mid \bm\xi)$, (b) $\mathbb{L}(c \zeta \mid \bm\xi) = c \mathbb{L}(\zeta \mid \bm\xi)$ for $c \in \mathbb{R}$, (c) $\mathbb{L}(\zeta \mid \bm\xi) = 0$ if $\Cov(\zeta,\bm\xi) = \bm 0$, (d) $\mathbb{L}(\zeta \mid \bm\xi) = \zeta$ if $\zeta \in \operatorname{Span}(\bm\xi)$, and (e) $\mathbb{L}(\zeta \mid \bm\xi) = \mathbb{L}(\zeta \mid \mathbf A \bm \xi)$ for invertible $\mathbf A$. Thus, $\mathbb L(\cdot\mid \star)$ mimics $\mathbb E(\cdot\mid \star)$, and Lemma \ref{lemma:linear2} holds.
The proof is similar to that of Lemma \ref{lemma:linear}.

\begin{lemma}\label{lemma:linear2}
    In \eqref{equation:model}, Lemma \ref{lemma:linear} holds with $\mathbb E(\cdot\mid \star)$ being replaced by $\mathbb L(\cdot\mid\star)$.
\end{lemma}

As a result, if $\bm X$ and $\bm\varepsilon$ are uncorrelated as in \eqref{equation:model}, the derivation in Section \ref{section:known-order-estimation} holds with $\mathbb E(\cdot\mid \star)$ being replaced by $\mathbb L(\cdot\mid\star)$.

\bibliographystyle{apalike}
\bibliography{ref}

\end{document}